\documentclass[aps,pra,reprint,amsfonts,twocolumn,floatfix,nofootinbib]{revtex4-2}
\usepackage{bm,graphicx,mathrsfs}
\usepackage{graphicx}
\usepackage{epsfig}
\usepackage{amsmath,bbm}
\usepackage{amsfonts,amssymb}
\usepackage{times}
\usepackage{color}
\usepackage{verbatim}
\usepackage{enumitem}
\usepackage{amsthm}

\usepackage{amssymb,mathtools,physics,graphicx,xcolor}
\usepackage{hyperref}

\usepackage[capitalize,nameinlink]{cleveref}
\usepackage[dvipsnames]{xcolor}
\hypersetup{
  colorlinks=true,
  linkcolor=violet,
  citecolor=violet,
  urlcolor=violet
}

\DeclareMathOperator{\ad}{ad}

\newcommand{\cL}{{\cal L}}

\newcommand{\cO}{\mathcal{O}}

\newcommand{\cM}{{\cal M}}

\newcommand{\cK}{\mathcal{K}}

\def\>{\rangle}
\def\<{\langle}

\newcommand{\be}{\begin{equation}}
\newcommand{\ee}{\end{equation}}
\newcommand{\bq}{\begin{eqnarray}}
\newcommand{\eq}{\end{eqnarray}}
\newcommand{\bea}{\begin{eqnarray}}
\newcommand{\eea}{\end{eqnarray}}

\newtheorem{theorem}{Theorem}
\newtheorem{lemma}[theorem]{Lemma}
\newtheorem{corollary}[theorem]{Corollary}

\newtheorem{definition}[theorem]{Definition}

\newenvironment{remark}{\vspace{1.5ex}\par\noindent{\it Remark:}}%
    {\hspace*{\fill}$\Box$\vspace{1.5ex}\par}

\begin{document}

\definecolor{james}{rgb}{1,.6,0}
\newcommand{\mjk}[1]{{\color{james} #1}}
\newcommand{\je}[1]{{\color{green} #1}}

\title{Reducing Circuit Depth in Lindblad Simulation 
via Step-Size Extrapolation}

\author{Pegah Mohammadipour}
\email{pegahmp@psu.edu}
\author{Xiantao Li}
\email{xiantao.li@psu.edu}
\affiliation{Department of Mathematics, \\ The Pennsylvania State University,\\  University Park, Pennsylvania 16802, USA}
\date{July 30, 2025}
\begin{abstract}
We study \emph{algorithmic error mitigation} via Richardson–style extrapolation for quantum simulations of open quantum systems modelled by the Lindblad equation.  
Focusing on two specific first-order quantum algorithms, we perform a backward-error analysis to obtain a step-size expansion of the density operator with explicit coefficient bounds.  
These bounds supply the necessary smoothness for analyzing Richardson extrapolation, allowing us to bound both the deterministic bias and the shot-noise variance that arise in post-processing.  
For a Lindblad evolution with generator bounded by $\ell$, our main theorem shows that an $n=\Omega \bigl(\log(1/\varepsilon)\bigr)$-point extrapolator reduces the maximum circuit depth needed for accuracy~$\varepsilon$ from polynomial $\mathcal{O} ((\ell T)^{2}/\varepsilon)$ to polylogarithmic $\mathcal{O} \bigl((\ell T)^2 (\log (\ell T)) \log^2(1/\varepsilon)\bigr)$ scaling, an exponential improvement in~$1/\varepsilon$, while keeping sampling complexity to the standard $1/\varepsilon^2$ level, thus extending such results for Hamiltonian simulations to Lindblad simulations. Several numerical experiments illustrate the practical viability of the method. 
\end{abstract}
\maketitle

\section{Introduction}
Simulating quantum systems is widely regarded as one of the central applications of quantum computing \cite{feynman1982simulating}. However, in practice, quantum systems are never perfectly isolated; they inevitably interact with their surrounding environment, leading to phenomena such as decoherence and dissipation. These effects necessitate the study and simulation of \textit{open quantum systems}, whose dynamics are no longer governed solely by the Schr\"odinger equation but instead require more general frameworks that account for environmental interactions.  The evolution of open quantum systems is typically described by \textit{quantum master equations} for the system’s density operator \(\rho(t)\), which in the Markovian regime, is governed by the \textit{Lindblad master equation}~\cite{lindblad1976generators, gorini1976completely},
\begin{equation}\label{eq:lindblad}
    \frac{d}{dt} \rho 
    = - {i}[H,\rho] + \sum_{j=1}^{J} \kappa_j \left( L_j \rho L_j^{\dagger} - \frac{1}{2} \{ L_j^{\dagger} L_j, \rho \} \right).
\end{equation}
 The Hermitian operator $H$ is the system Hamiltonian, which induces the unitary (coherent) part of the evolution. The operators $L_j$, called \textit{jump operators}, model different types of interactions between the system and its environment, such as energy loss, decoherence, or particle exchange. Each $\kappa_j \geq 0$ is a rate parameter that quantifies the strength of the corresponding interaction. Throughout the paper, we will absorb them into $L_j$ and thus set $\kappa_j=1$.

As interest in modeling open quantum systems grows, so too does the need for efficient quantum simulation methods. Numerous quantum algorithms have recently been proposed to simulate open-system dynamics with potential exponential speedups in system dimension \cite{KBG11,CL17,cleve2016efficient,li2022simulating,schlimgen2022quantum,chifang2023quantum,patel2023wave1,patel2023wave2,pocrnic2023quantum,Ding2024,kamakari2022digital,kato2025exponentiallyaccurateopenquantum,chen2025randomizedmethodsimulatinglindblad,yu2024exponentiallyreducedcircuitdepths}.
Despite recent progress, circuit depth remains the dominant practical bottleneck for simulating open quantum systems on current noisy intermediate-scale quantum (NISQ) devices.

To approximate the system's evolution over a total time $T$, most algorithms divide the dynamics into smaller intervals, i.e., time steps, and apply a quantum circuit for each one. As the number of time steps increases, so does the total number of quantum gates. This deepens the circuit and accumulates hardware imperfections such as gate errors and stochastic noise. As a result, the simulated quantum noise may no longer reflect the intended physical model but rather the limitations of the device itself.

This issue parallels that of Hamiltonian (closed-system) simulation, where the quantum circuit approximates the continuous evolution $e^{-iHt}$ using discrete gate sequences such as Trotter–Suzuki decompositions. For example, in a first-order Trotter scheme, the circuit depth required to simulate time $T$ with error no more than $\varepsilon$, besides a commutator factor, scales as $\mathcal{O}(T^2/\varepsilon)$. The parameter $\varepsilon$ here quantifies the desired \textit{simulation precision}, i.e., how closely the implemented circuit approximates the ideal evolution.

A promising near-term remedy is \textit{algorithmic error mitigation}, which has been proposed and empirically validated for Hamiltonian simulations ~\cite{Endo2019,rendon2024improved,watson2024exponentially,mohammadipour2025direct} using the Trotter algorithm as the base method.  The strategy is elegantly simple:
\begin{enumerate}
    \item Run a low-order algorithm (e.g., first-order Trotter) at several different coarse time step sizes \(\tau_1 < \tau_2 < \cdots < \tau_n\), keeping each quantum circuit shallow.
    \item Combine the resulting expectation values using classical post-processing, typically via Richardson extrapolation or polynomial fitting, to cancel leading-order discretization errors.
\end{enumerate}

The classical extrapolation step requires only \({\mathcal{O}}(n^3)\) floating-point operations. Yet, it effectively eliminates the dominant step-size error term without increasing the quantum depth, replacing the typical inverse-polynomial depth scaling \(\mathcal{O}(T^{1+1/p}/\varepsilon^{1/p})\) with a polylogarithmic dependence on the target precision \(\varepsilon\) \cite{watson2024exponentially}. This trade-off is especially attractive for NISQ devices, where coherence time is limited and quantum depth is at a premium.

Formally, for any observable $O$ we define
\begin{equation}\label{f-def}
f_\tau(T) ={\Tr} \bigl(\rho_{\tau}(T)\,O\bigr),    
\end{equation}
where $\rho_{\tau}(T)$ is the density operator produced by a step-size-$\tau$ integrator after time $T$. 

The dependence on the stepsize is often interpreted from an expansion,
\begin{equation}\label{f-expansion}
f_\tau(T)=f(0)+\alpha_1\tau^{p}+\alpha_2\tau^{p+1}+\cdots,
\end{equation}
where $f(0)$ is the true observable value in the zero-step-size limit, and the $\alpha_1\tau^{p}$ represents a leading error term of order $p.$

The postprocessing step constructs coefficients $\{\gamma_i\}_{i=1}^{n+1}$ such that
\begin{equation}\label{gamma_j2f}
    f_{\text{extrap}}(0) = \sum_{i=1}^{n+1} \gamma_i f_{\tau_i}(T),
\end{equation}
Each $f_{\tau_i}$ is an empirical mean of circuit-level measurements, while the $\gamma_i$ are determined, for instance, by an interpolation using a polynomial via solving a Vandermonde system or a least-squares fit (regression) to a polynomial of lower degree \(m \leq n\). In either case, the leading term in the error in \Cref{f-expansion} will be improved to $p+m-1$.  No additional quantum resources are required, yet the dominant step-size error can be suppressed by $m-1$ orders, thus dramatically reducing the maximum circuit depth needed to reach precision $\varepsilon$.  

Despite the intuitive appeal of Richardson–style extrapolation, its practical efficacy hinges on two subtle issues.   
First, the error expansion \eqref{f-expansion} is valid only when the observable map \(f(\tau)\) is sufficiently smooth in \(\tau\); the derivative bounds depend on the chosen algorithms and must be established rigorously for each algorithm.  
Second, each empirical estimate \(f(\tau_i)\) in the linear combination~\eqref{gamma_j2f} is subject to statistical fluctuations due to finite sampling (shot noise). As the number of grid points increases, the variance of the extrapolated result may also grow, potentially offsetting the benefits of bias reduction.

In this work, we perform a full bias–variance analysis for two Lindblad-simulation primitives and prove that extrapolation still delivers a net benefit even after accounting for statistical fluctuations.  
Our results extend the extrapolation method for Hamiltonian simulations to the inherently irreversible dynamics of open quantum systems and provide the complexity guarantees for error-mitigated Lindblad evolution.

\paragraph{Contributions.}
The main contributions of this paper are as follows:
\begin{enumerate}
    \item We analyze two specific quantum algorithms for simulating Lindblad dynamics \cite{cleve2016efficient,Ding2024} and derive rigorous upper bounds on both the deterministic bias and the statistical error introduced by Richardson–style extrapolation.
    \item To obtain an end‑to‑end complexity estimate for the post‑processing scheme, we carry out a backward‑error analysis that reveals how the approximate density operator produced by each algorithm depends on the time step~$\tau$. We prove that the extrapolated observables remain Gevrey-smooth in the step size, even under non-unitary Lindblad dynamics.  This result is the key ingredient in proving the circuit‑depth and sampling‑complexity guarantees of the extrapolation method. 
    \item We show that extrapolation reduces the maximum circuit depth for Lindblad simulation from $\frac{(lT)^{2}}{\varepsilon}$ to $(lT)^{2} (\log l) \log^2\bigl(1/\varepsilon\bigr)$, yielding an exponential improvement with respect to the target accuracy~$\varepsilon$.
\end{enumerate}

\paragraph{Related work.}  
Extrapolating physical error rates, usually modeled by Lindblad equations as well, has become a standard noise-mitigation technique on NISQ devices \cite{KTemme_SBravyi_JGambetta,Kandala_ZeroNoise_2019,endo2018practical,digital_zero_noise_extrap}. The same idea has also been applied to Trotterized Hamiltonian simulation \cite{Low2019,Vazquez2023,Endo2019,Rendon2024,watson2024exponentially,zhuk2024trotter}, where the extrapolation variable is the Trotter step size instead of the physical error rate. Recent works even combine the two, simultaneously extrapolating algorithmic and physical errors \cite{Endo2019,mohammadipour2025direct,hakkaku2025data}.

Our analysis is closest to that of \cite{Rendon2024,watson2024exponentially}, but we target open-system (Lindblad) dynamics. In particular, \cite[Lemma 4]{watson2024exponentially} establishes how the observable-estimation error scales with the Trotter step and, crucially, shows that Richardson or Chebyshev extrapolation can reduce circuit depth without worsening the dependence on the total simulation time $T$. We prove an analogous result for Lindblad evolutions via backward-error analysis. A key difference is that higher-order Suzuki–Trotter type formulas for closed-system simulation necessarily involve negative step sizes \cite{CL17}; such negative intervals violate complete positivity in Lindblad dynamics. Consequently, we restrict attention to low-order splitting formulas, making our method a genuine extrapolation to $\tau\to 0$ rather than an interpolation across positive and negative steps.

\paragraph{Outline and Notation.}
Throughout this paper, we will work with the density matrix and related matrices. Their magnitude will be measured by the trace norm (or Schatten $1$-norm), denoted by $\norm{A}$.  Also involved are super operators, acting on Hermitian matrices. 
Given a superoperator $\mathcal{M}$,  the induced norm is given by
\begin{equation}
\norm{\mathcal{M}} =\sup_{\norm{A}\leq 1} \norm{ \mathcal{M}(A)}.
\end{equation}

One important example is the Lindblad operator, 
\begin{equation}
    \mathcal{L} \rho := - i[H,\rho] + \sum_{j=1}^{J}  \left( L_j \rho L_j^{\dagger} - \frac{1}{2} \{ L_j^{\dagger} L_j, \rho \} \right).
\end{equation}
whose norm can be directly bounded by,
\begin{equation}\label{l-param}
    \|\mathcal{L}\| \leq \ell, \quad \ell:= 2 \norm{H} + 2 \sum_{j=1}^{J} \norm{L_j}^2.
\end{equation}

\section{Background} \label{sec:background}
The main purpose of this paper is to approximate the solution to the Lindblad master equation \eqref{eq:lindblad}. We consider the evolution of the density operator \(\rho(t)\) until the time \(T\) according to the Lindblad equation \eqref{eq:lindblad}. 
To approximate the solution of the Lindblad equation numerically, we rewrite \eqref{eq:lindblad} as the following initial-value problem:  
\begin{equation}
    \frac{\partial}{\partial t} \rho (t) = \mathcal{L} \rho(t), \quad \rho(0)=\rho_0, 
\end{equation}
with the exact solution:
\[
\rho(t)=e^{t \mathcal{L}}\rho_0.
\]

In practice, computing the full exponential of \(\mathcal{L}\) is often infeasible, so we discretize the continuous-time evolution by employing a one-step numerical scheme to approximate the evolution.

Let $\tau > 0$ be the time-step size, and define the discrete approximation to $e^{\tau \mathcal{L}}$ by a linear map $\mathcal{K}(\tau)$, which serves as the \textit{discrete solution operator}. That is,
$$
e^{\tau \mathcal{L}}  \approx \mathcal{K}(\tau).
$$

Given this approximation, the state at time $t_n = n\tau$ is computed iteratively as:
\begin{equation}
\rho_n = \mathcal{K}(\tau)^{n} \rho_0,
\end{equation}
where $\rho_n \approx \rho(t_n)$ is the numerical estimate of the density operator at time $t_n$, and $\mathcal{K}(\tau)^{n}$ denotes $n$ successive applications of $\mathcal{K}(\tau)$. The choice of $\mathcal{K}(\tau)$ depends on the specific simulation method used—for instance, it may arise from a Kraus operator expansion, a dilation-based approximation of the Lindblad dynamics, or a Trotter product formula.

We can expand the discrete solution operator $\mathcal{K}$ 
\begin{equation} \label{eq:kappa_expansion}
\mathcal{K}(\tau)=I+\tau \mathcal{M}_1+\tau^2 \mathcal{M}_2+\tau^3 \mathcal{M}_3+\cdots,
\end{equation}
where \(\mathcal{M}_1 = \mathcal{L}\). This is due to the fact that as $ \tau \to 0$, the approximate solution has the same generator as the exact dynamics. This is often referred to as consistency in ODE solvers \cite{deuflhard2012scientific}.

We observe that the numerical solution exhibits an expansion with powers of the step size. 
To improve the accuracy of observable estimates in time-discretized Lindblad simulations, we apply \textit{Richardson extrapolation} \cite{sidi2003practical} 
to the observables in \Cref{f-def}.
Specifically, we can evaluate the observable at multiple step sizes \(\tau_j = T / n_j\). A polynomial interpolation or fitting is given by, 
\begin{equation}\label{pn-extrap}
p_n(\tau) = \sum_{j=1}^{n+1} \gamma_j(\tau) f_{\tau_j}. 
\end{equation}

For example, when polynomial interpolation is employed, the coefficients \(\gamma_j\) are determined by
\(
\sum_{j=0}^n \gamma_j \tau_j^m = 0  \text{ for } m = p, \dots, p+n-1, \text{ and } \sum_{j=0}^n \gamma_j = 1.
\)
These algebraic conditions help to eliminate the $n$ leading error terms, and improve the error to \(\mathcal{O}(\tau^{p+n})\).
 Aside from polynomial interpolation, regression methods can be used as well, for which the analysis is similar \cite{Trefethen_approx}.

The performance of extrapolation heavily depends on the smoothness properties of $f$ in \Cref{f-def}, which in turn depends on the expansion of $\mathcal K(\tau)$ in \Cref{eq:kappa_expansion}. In the context of Suzuki-Trotter algorithms for Hamiltonian simulations, such expansion was obtained through an expansion of the shadow Hamiltonian \cite{rendon2024improved}.  In this paper, we use backward error analysis from numerical analysis \cite{deuflhard2012scientific} and identify the equations governing the coefficients in the error expansion. We summarize the results as follows.

\begin{lemma} \label{lemma:discretization_expansion}
Assume that the approximation method has an expansion in \Cref{eq:kappa_expansion}.
Let \(\rho_\tau (t)\) denote the corresponding approximate solution at time \(t\) with stepsize \(\tau\).  There exists a sequence of smooth functions \(\Gamma_k(t)\) such that
\begin{equation} \label{eq:rho_expansion}
    \rho_\tau(t) = \rho(t) + \tau \Gamma_1 (t) + \tau^2 \Gamma_2(t) + \dots ,
\end{equation}
where \(\rho(t)\) is the solution to the exact evolution \eqref{eq:lindblad}. For a discrete solution  \eqref{eq:kappa_expansion}, these coefficient matrices satisfy the initial value problems \(\Gamma_k(0) = 0\) for all \(k \geq 1\), and evolution equations,
\begin{equation} \label{eq:gamma_ode}
    \Gamma'_{k-1}(t) = \mathcal{L} \Gamma_{k-1}(t)  - \frac{\mathcal{L}^k \rho (t)}{k!} + \sum_{i=2}^k \left( \mathcal{M}_i \Gamma_{k-i}(t) -  \frac{\Gamma^{(i)}_{k-i} (t)}{i!} \right).
\end{equation}
\end{lemma}

The expansion of the numerical solution with respect to the stepsize follows from standard backward error analysis \cite{deuflhard2012scientific}. Our analysis considers the Lindblad dynamics and provides an explicit form of the governing equations for the coefficients in the error expansion \eqref{eq:rho_expansion}. For a detailed proof, refer to  \Cref{proof:coefficient_first_order}.

As concrete examples, we consider two approximation methods:  Kraus form approximation for the channel induced by the Lindblad equation \eqref{eq:lindblad} and a dilation-based approximation expressed in a Stinespring form. With these two specific approximations, we will derive the expansion of the corresponding $\cK(\tau)$ and therefore the error expansion in \Cref{eq:rho_expansion} by analyzing \Cref{eq:gamma_ode}. Importantly, we will use these bounds to analyze the improvement brought forth by the extrapolation methods.

\section{First-order Approximation by a Kraus Form} One specific construction of the discrete solution operator \(\mathcal{K}\) is to regard $e^{\tau \mathcal L}$  as a quantum channel and build the approximation in Kraus form \cite{cleve2016efficient}. We define Kraus operators as follows, 
\begin{equation} \label{eqs:Kraus_operators}
    F_0 = I + \left(-iH - \frac{1}{2}\sum_j L_j^\dagger L_j\right)\tau, \quad F_j = L_j\sqrt{\tau},    1\leq j\leq J,
\end{equation}
Then, the evolution from \(\rho_n\) to \(\rho_{n+1}\) can be expressed in the \textit{Kraus form}:
\begin{equation}\label{eqn:Kraus_form}
\rho_{n+1} =  \mathcal{K}[\rho_n] = F_0 \rho_n F_0^\dagger + F_1 \rho_n F_1^\dagger + \cdots + F_J \rho_n F_J^\dagger.
\end{equation}

One can verify that the scheme in \eqref{eqn:Kraus_form} constitutes a \emph{first-order} numerical method for simulating Lindblad dynamics \cite{cleve2016efficient}. That is, for sufficiently small step size $\tau$, the channel $\mathcal{K}$ approximates the exact evolution $e^{\tau \mathcal{L}}$ with an error that scales  $\ell^2 \tau^2$. More precisely, for a fixed initial state $\rho$, the local error satisfies
\begin{equation}
\left\| \mathcal{K}[\rho] - e^{\tau \mathcal{L}}[\rho] \right\| = \mathcal{O}(\ell^2\tau^2),
\end{equation}
which implies that the global error after $N = T/\tau$ steps accumulates as $\mathcal{O}(\tau)$. This convergence rate is characteristic of first-order integrators, such as the forward Euler method in classical numerical analysis. The construction in \eqref{eqs:Kraus_operators} can thus be interpreted as a completely positive trace-preserving (CPTP) analogue of the Euler method tailored for open quantum systems.

The algorithms in \cite{cleve2016efficient} incorporate a compression scheme to achieve higher-order accuracy. However, this approach relies on sophisticated logical gate constructions, which we do not employ here, as our focus is on near-term implementability. Moreover, due to the explicit form of the Kraus operators, each being at most linear in $\tau$, the expansion of $\mathcal{K}$ in \eqref{eq:kappa_expansion} and \eqref{eqn:Kraus_form} satisfies $\mathcal{M}_i = 0$ for all $i \geq 3$, and
\begin{equation}
\norm{\mathcal{M}_2} \leq \norm{F_0 - I}^2.
\end{equation}

We observe that this is effectively the norm of the non-Hermitian Hamiltonian. Furthermore, using \Cref{eqs:Kraus_operators}, we can refine the bound:
\begin{equation}\label{M2-bound}
\norm{\mathcal{M}_2 } \leq \left(\norm{H}  + \frac{1}{2} \sum_{j} \norm{L_j}^2 \right)^2 = \Theta(\ell^2).
\end{equation}

Using the Kraus form above, we derive bounds on the coefficients in the error expansion \eqref{eq:rho_expansion}. To obtain explicit estimates, we introduce a generating sequence that characterizes the magnitude of the terms $\Gamma_k$ and their derivatives appearing in \eqref{eq:rho_expansion}.

\begin{definition}[Generating sequence for Kraus Form] \label{def:c_s_kraus}
We define \(c_{i,j,k}\) as non-negative real numbers such that \( c_{0,0,k} = 0 \text{ for } k \geq 1\), \(c_{i,j,0}= \delta_{j,0} \ell^i\) for \(i, j \ge 0\), and the rest of the entries are generated from the recursion relations below:
\begin{widetext}
\begin{equation}
\begin{aligned}
   c_{0,j,k} = & \, \frac{B}{j}  \cdot c_{0,j-1,k-1} + \delta_{j,1} \cdot \frac{\ell^{k+1}}{(k+1)!}
   + \sum_{p=1}^{k -j} \frac{c_{p+1,j-1,k-p}}{j(p+1)!}, \, j=1,2,\dots,k \\
   c_{i,j,k} = & \, \ell \cdot c_{i-1,j,k} +  B \cdot c_{i-1,j,k-1}
   + \delta_{j,0} \cdot \frac{\ell^{i+k}}{(k+1)!} + \sum_{p=1}^{k-j} \frac{c_{i+p,j,k-p}}{(p+1)!},
   \, j=0,1,\dots,k-1 \\
   c_{i,k,k} = & \, \ell \cdot c_{i-1,k,k}.
\end{aligned}
\end{equation}
\end{widetext}
where $\ell$ is the bound on the Lindblad generator \Cref{l-param}, which without loss of generality is assumed to be larger than 1, and \(B := \norm{\cM_2} = \Theta( \ell^2)\) by \eqref{M2-bound}.
\end{definition}

The elements of the sequence can be generated recursively using the definition above. Moreover, with an inductive argument, we can establish the following bounds on the growth of the sequence.

\begin{lemma} \label{lemma:c_bound_kraus}
    Let \(c_{i,j,k} \geq 0\) be as in \textit{Definition} \ref{def:c_s_kraus}. Then for all \(i,j,k\ge0\),
    \begin{equation} \label{ineq:c_bound}
        c_{i,j,k}   \le   \frac{C^{i+k}_1 \, C^{k}_2}{j!}.
    \end{equation}
    where 
      \begin{equation}\label{def: C1}
      C_1 := \max \{ B, \ell (e+1), 1 \}, \text{ and  } \, C_2 \geq (e+1) \log (C_1)
  \end{equation}
\end{lemma}

The proof is presented in \Cref{app:c_bound}. This result is important because it provides explicit bounds on the terms \(\Gamma_k\) and their derivatives appearing in \eqref{eq:rho_expansion}, as will be demonstrated in the next lemma. In particular, this yields bounds in the Gevrey class, setting the stage for controlled extrapolation.

\begin{lemma} \label{lemma:gamma_bound}
The coefficients in the expansion  of the density operator \(\rho_{\tau}\) in \Cref{lemma:discretization_expansion}, approximated by the Kraus operator \eqref{eqn:Kraus_form}, satisfy the following bound: 
\begin{widetext}
\begin{equation}
    \bigl\|\Gamma_{k}^{(i)}(t)\bigr\|
    \le  
  P_{i,k}(t)
    =   
  \sum_{j=0}^k
  c_{i,j,k}\,t^{j} \, \text{ for } i\geq 0, k \geq 1, \quad  c_{0,0,k} = 0 , \quad c_{i,j,0}= \delta_{j,0} l^i,
\end{equation}
\end{widetext}

where $\Gamma_k^{(0)} = \Gamma_k$, $\Gamma_k^{(1)} = \Gamma'_k$, and the coefficients \(c_{i,j,k}\) are defined by the generating  sequence in \textit{Definition} \ref{def:c_s_kraus}.
\end{lemma}
 We provide a detailed proof of this lemma in \Cref{app:gamma_bound}. 
\\
In simulating the Lindblad dynamics \eqref{eq:lindblad}, one can assume without loss of generality that the final evolution time $T\leq 1$. Time integrations for longer time periods can be accomplished by rescaling the right hand side, i.e., $\cL$, accordingly. We note that when $T \geq 1$, one can apply a change of variables \(s= t/T\), and set 
 \begin{equation}
      \hat{\rho}(s) = \rho(s T) , \quad 0 \le s \le 1,
 \end{equation}
 which follows the Lindblad equation,
 \begin{widetext}
      \begin{equation}\label{rhorhohat}
     \dfrac{d}{d\tau}  \hat{\rho}(s) =  \dfrac{d}{dt} \rho(s T) = T \cL \rho(sT )  = T \cL   \hat{\rho}(s),  \,  \hat{\rho}(0)=\rho_0, \,  \hat{\rho}(1) = \rho(T).
 \end{equation}
 \end{widetext}
We then simulate $ \hat{\rho}$ up to time $s=1$, with the operator norm of the Lindblad operator $\ell$ multiplied by $T$.  Therefore, $\ell T$ will appear together in the simulation complexity. This scaling for a first-order method ($p=1$) is also consistent with that in \cite[Sec. 2.3]{watson2024exponentially}.

\begin{corollary} \label{cor:gamma_bound_1}
     For $t\leq 1$, the coefficients in the expansion  of the density operator \(\rho_{\tau}\) in \Cref{lemma:discretization_expansion}, approximated by the Kraus operator \eqref{eqn:Kraus_form}, satisfy the following bound for every \(i,k\in\mathbb N\),
  \begin{equation}
       \norm{\Gamma^{(i)}_k(t) }     \le  
     e\,C_1^{\,i+k}\,C_2^{\,k},
  \end{equation}
  where \( C_1 := \max \{ B, \ell (e+1), 1 \}, \text{ and  } \, C_2 \geq (e+1) \log (C_1)\).
\end{corollary}

For a proof of the corollary, we refer the reader to \Cref{proof:gamma_bound}. These bounds enable us to derive upper bounds on the expectation value of any bounded observable, as well as on all its derivatives with respect to the time step size \(\tau\), as follows.

The above corollary helps us find the radius of convergence of the series \(\rho_{\tau}\) in Lemma \ref{lemma:discretization_expansion}, i.e. \(\tau_{\max}\), which we introduce in the following theorem.

\begin{theorem} \label{thm:Gevrey-f}
    Let $O$ be a bounded observable, and 
  \(
      f(\tau)  :=  \tr   \bigl(\rho_\tau(T)\,O\bigr),
     \, 0 \le T \le 1, 
  \)
  where $\rho_\tau(T)$ is approximated using the Kraus form \eqref{eqn:Kraus_form}. Then $f(\tau)$ belongs to the Gevrey class on  $[0,\tau_{\max}]$ and satisfies the following bound,
    \begin{equation} \label{ineq:fk-bound}
        \abs{f^{(k)}(\tau)} \leq \sigma \, \nu^k \, k!, \quad \, \qquad \text{for } \tau \in [0,\tau_{\max}], \quad k \ge 1,
    \end{equation}
    where \(\sigma := 2e \|O\|\) and \(\nu := 2C_1 C_2\). If \(T \le 1\), then \cref{ineq:fk-bound} holds for any \(\tau_{\max} \le 1/(2\nu)\). When \(T \ge 1\), the scaling in \cref{rhorhohat} applies and the bound remains valid with \(\nu := 2C_1 C_2\, T^2 \log (\ell T)\).
\end{theorem}
The proof will be presented in \Cref{app:bias_bound}. 

The scaling in \cref{rhorhohat} changes $\ell$ to $\ell T$, and thus for large $T$, which is the regime of interest when circuit depth reduction is the main focus, the dominating term in \cref{def: C1} for $C_1$ is $\ell^2 T^2$. Therefore, from this point onward we assume \(T \gg 1\). The $\tau$–series for $\rho_\tau(t)$ is therefore real-analytic at $\tau=0$ with radius
\begin{equation} \label{eq:t_max_bound}
\tau_{\mathrm{max}} \;\le\; \frac{1}{2\nu}
\;=\;
\begin{cases}
\Theta\!\bigl((\ell^2 \log \ell)^{-1}\bigr), & T < 1,\\[2mm]
\Theta\!\bigl((\ell^2 T^2 \log (\ell T))^{-1}\bigr), & T \ge 1.
\end{cases}
\end{equation}
This behavior mirrors that of Hamiltonian simulations, where errors scale with powers of $\tau\|H\|$, the dimensionless product of the simulation time step-size $t$ and the Hamiltonian's operator norm $\|H\|$.

 To improve the approximation of the exact expectation value \(f(0)\), we apply polynomial extrapolation using function evaluations at a sequence of step sizes \(\tau_j \in (0,\, \tau_{\max})\), for \(j = 1, 2, \ldots, n+1\). While this strategy enhances accuracy near \(\tau = 0\), it inevitably introduces an extrapolation error. The Gevrey-class bound \eqref{ineq:fk-bound} established above, provides a precise tool for analyzing and controlling this extrapolation error.

\textit{Remark:}\label{lemma:bias_bound_richardson_equi}
Let \( f \) be as in \textit{Theorem~\ref{thm:Gevrey-f}}, and let \( p_n(\tau) \) denote the degree-\(n\) polynomial interpolant of \( f \) at the equidistant nodes \(\{ \tau_j = jh \}_{j=1}^{n+1} \), where \( h := \tau_{\max}/(n+1) \). Then, for any desired extrapolation accuracy \(\varepsilon > 0\), if
\begin{equation} \label{n-lower-equid}
    n   \ge   2 \log \left( \frac{\sigma \sqrt{2\pi}}{\varepsilon} \right),
\end{equation}
we obtain the approximation guarantee
\(
    \bigl|f(0) - p_n(0)\bigr| < \varepsilon,
\)
which is calculated using the classical exponential convergence of analytic functions under interpolation at points~\cite{Trefethen_approx}. This shows that extrapolation using equidistant nodes can, in principle, achieve \(\varepsilon\)-precision with only logarithmic scaling in the number of interpolation points.

However, as observed in prior work~\cite{watson2024exponentially, cai2021multi}, each term in the extrapolated estimator
\begin{equation} 
    p_n(0) = \sum_{j=1}^{n+1} \gamma_j f(\tau_j)
\end{equation}
must be independently sampled and is subject to statistical error. The magnitude of this noise is controlled by the quantity \(\|\gamma\|_1 := \sum_{j=1}^{n+1} |\gamma_j|\), which coincides with the \textit{Lebesgue constant} of the interpolation process~\cite{Trefethen_approx}. For equidistant nodes, the Lebesgue constant grows exponentially with \(n\) as
\begin{equation} \label{gamma1-equid}
    \|\gamma\|_1 = 2^{n+1} - 1,
\end{equation}
as shown in classical interpolation theory~\cite{brutman1996lebesgue}. Combining this with the lower bound \eqref{n-lower-equid} on \(n\), the resulting \textit{variance amplification} in the estimator, assuming fixed shot noise per circuit, scales as \(\mathcal{O}(1/\varepsilon^4)\), which renders the method impractical for high-precision settings. Thus, despite the favorable bias error from extrapolation, the large variance due to equispaced interpolation negates any advantage over direct estimation.

Successful extrapolation methods are typically based on Chebyshev nodes due to their superior stability and approximation properties. In our setting, we adapt this strategy by mapping Chebyshev nodes to lie within the interval \([0,\, \tau_{\max}]\), and proceed to analyze the resulting extrapolation scheme. To this end, we begin by stating two standard lemmas from polynomial approximation theory~\cite{Trefethen_approx}.

\begin{lemma}\label{lemma:bias_bound_richardson_cheby}
Let \( f \) be as in \textit{Theorem~\ref{thm:Gevrey-f}}, and let \( p_n(\tau) \) denote the degree-\(n\) polynomial that interpolates \( f \) at the Chebyshev nodes of the first kind mapped to \([0,\, \tau_{\max}]\):
\begin{equation}\label{tauk_cheb}
    \tau_k = \tau_{\max} \left( 1 - \cos \theta_k \right)/2,  \  k = 1, \dots, n+1, 
\end{equation}
where \(\theta_k = \frac{2k - 1}{2(n+1)} \pi\). Then, for \( n \geq \frac{1}{3 \log 2} \log\left( \sigma/\varepsilon \right) - \frac{2}{3} \), we have
\(
|f(0) - p_n(0)| < \varepsilon.
\)
\end{lemma}
The result follows by applying the definition of Chebyshev nodes, mapping them to the desired interval, and using the interpolation error formula, all of which are available in \cite{Trefethen_approx}.

\begin{lemma}\cite[Theorem 15.2]{Trefethen_approx} \label{lemma:lagrange_bound_chebyshev}
 Let $\{\tau_j = (\tau_{\max} \left( 1 - \cos \theta_j \right)/2\}_{j=1}^{n},$ where \(\theta_j = \frac{2j - 1}{2n} \pi\), be the Chebyshev nodes of the first kind mapped to \([0,\, \tau_{\max}]\), then the sum of the absolute values of the Lagrange basis functions evaluated at zero, \(\gamma_j = L_j(0)\), is bounded logarithmically:
    \begin{equation}
        \sum_{k=1}^{n} |\gamma_k| = \cO(\log n).
    \end{equation}
\end{lemma}

The importance of the logarithmic bound of $\norm{\gamma}_1$ can be appreciated from the Hoeffding bound,
\begin{equation}\label{eq: hoeffding}
\displaystyle \mathbb{P}(|S_{N_S} - \mathbb{E}(S_{N_S})| \geq \varepsilon) \leq 
2 \exp(- \frac{\varepsilon^2 N_S}{2 \alpha^2 (\displaystyle \sum_{j=1}^{n+1}  \abs{\gamma_j} )^2}).    
\end{equation}

Equivalently, to ensure accuracy \(\varepsilon\) with failure probability at most \(\delta\), we require $N_S$ samples
\begin{equation}
N_S   \ge   \frac{2 \alpha^2 \|\gamma\|_1^2}{\varepsilon^2} \log \frac{2}{\delta}.
\end{equation}
While the exponential dependence in \Cref{gamma1-equid} will lead to higher sampling complexity, a polynomial or a logarithmic scaling with $n$ will retain a standard  $\cO(\varepsilon^{-2})$ complexity.

In any digital simulations of the Lindblad dynamics, the step size \(\{\tau_j\}\) 
must be chosen such that the number of steps is an integer.  This practical constraint means we cannot use the ideal Chebyshev nodes directly and must instead use nearby 'quantized' time steps, the effect of which must be taken into account.
Intuitively, to ensure that the perturbed nodes remain distinct and well-conditioned, the perturbation magnitude \(|\tau_j - \xi_j|\) must be smaller than the minimal spacing between Chebyshev nodes. From \textit{Theorem~\ref{thm:Gevrey-f}}, this spacing is of order \(\cO(\tau_{\max} / n^2)\), which provides a bound on the allowable perturbation size. The next lemma investigates the robustness of the Chebyshev node configuration by analyzing how perturbations affect the associated Lebesgue constant. Although the robustness of Chebyshev interpolation is well studied in numerical analysis~\cite{vianello2018stability}, we provide a self-contained proof in the Appendix tailored to this discrete, integer-aligned perturbation. Unlike the analysis in~\cite{Low2019}, our argument does not rely on a Taylor expansion that neglects higher-order terms.

\begin{lemma}{Effect of Perturbed Chebyshev Nodes on Variance Bound}.
\label{lem:integer-reciprocal}
Let \(n \ge 2\) and \(\tau > 0\), and define Chebyshev nodes on \([0, \tau]\) by
\begin{equation}
  \xi_j = \frac{\tau}{2} \left(1 - \cos\left(\frac{2j - 1}{2n + 2}\pi\right)\right),
  \qquad j = 1, \dots, n + 1.
\end{equation}
If the time parameter $\hat{T}>\pi^{2}\tau n^{2}$, define the perturbed nodes by
\[
  k_j:=\bigl\lceil  \hat{T}/\xi_j\bigr\rceil,\qquad
  \tau_j:=\frac{\hat{T}}{k_j},\qquad j=1,\dots,n+1.
\]
Then the following statements hold:
\begin{enumerate}
  \item[\textup{(i)}] The perturbed nodes are strictly ordered: \(0<\tau_{1}<\tau_2<\dots<\tau_{n+1}<\tau\).
  \item[\textup{(ii)}] \(k_1>k_2>\dots>k_{n+1}\) are pairwise distinct positive integers.
  \item[\textup{(iii)}] 
        There exist constants \(C_1,C_2>0\) (independent of \(n\) and \(\hat{T}\))
        such that
\begin{equation}
     \sum_{j=1}^{n+1}\abs{\gamma_j}
          \leq
          C n^{4/(\pi^2-4)} \log n.
\end{equation}
\end{enumerate}
Furthermore, if $  \hat{T} >2 \tau\,n^{2} \log n$, then $ \sum_{j=1}^{n+1}\abs{\gamma_j}
          = \cO(\log n). $
\end{lemma}

We provide the proof in \Cref{proof:integer-reciprocal}. In view of \Cref{eq:t_max_bound}, we set \(\hat T = 1\) and
\(
  \tau := \tau_{\max}.
\)
Under the lemma’s threshold \(\hat T>\pi^{2}\tau n^{2}\), this is equivalent to
\[
  (\ell T)^2 \log (\ell T)   >   \pi^{2}\,n^{2}.
\]
If we instead enforce the stronger condition \(\hat T>2\,\tau\,n^{2}\log n\), then
\begin{equation}\label{T-vs-n-cond}
  (\ell T)^2 \log (\ell T)   >   2 \,n^{2}\log n.
\end{equation}
In what follows, we adopt this latter scaling. Combining it with the target precision \(\varepsilon\) (and the decay prescribed by Theorem~\ref{thm:Gevrey-f}), and using Lambert W function, we find that it is sufficient to choose
\begin{equation}\label{T>T*}
  T   \ge   \frac{ \sqrt{2} \log(1/\varepsilon)\,\sqrt{\log\log(1/\varepsilon)} }{\,\ell\,\sqrt{\log \ell}\,} .
\end{equation}

Combining the Gevrey-class smoothness from \cref{thm:Gevrey-f}, with the variance control for perturbed Chebyshev nodes from \cref{lem:integer-reciprocal}, and the Hoeffding bound \Cref{eq: hoeffding}, we arrive at the following end-to-end complexity guarantee:

\begin{theorem}\label{main-thm-kraus}
Given $\varepsilon>0$, assume that the simulation time satisfies \cref{T>T*}. 
     The Richardson extrapolation $p_n(\tau)$ of \(f(\tau)\) using $n+1$ Chebyshev nodes, with each expectation value $f(\tau_j) $, generated from the first-order approximation \eqref{eqn:Kraus_form} and  sampled with $N_S$ shots, is guaranteed to produce an estimator 
     \begin{equation}
         \mathbb P\left( \abs{p_n(0) - f(0) } < \varepsilon \right) > 1 - \delta,
     \end{equation}
     provided that $n = \Omega(\log \frac{1}{\varepsilon})$, $T= \Omega\left( \log \frac{1}{\varepsilon} \sqrt{\log \log \frac{1}{\varepsilon}} \right)$, and $N_S= \Omega \left( \frac{1}{\varepsilon^2}\log \frac{1}{\delta} \right). $
     In addition, to prepare $\rho(T)$ (i.e., $ {\hat{\rho}}(1) $ from \eqref{rhorhohat}), the maximum circuit depth of is $d_{\max} = \frac{1}{\tau_1}= \frac{T^2}{\tau_{\max}} n^2 = (\ell T)^2 (\log (\ell T)) \log^2(1/\varepsilon). $
\end{theorem} 

In contrast, the circuit depth without extrapolation, in order for the error $\ell^2 T \tau $ to be less than $\varepsilon  $, has to be $T/\tau = \cO( (\ell T)^2/\varepsilon)$. Therefore, our bound is exponentially better in $\varepsilon$ and remains the same in $T$.

\begin{remark}
 The results above can be readily extended to polynomial regression using Chebyshev nodes. In this case, the coefficients $\gamma_j$ can be expressed through Chebyshev polynomials as well. 
 Further, the number of nodes scales logarithmically with $\varepsilon$ to control the bias to be within $\varepsilon$ \cite{Trefethen_approx,mohammadipour2025direct}.  For the variance, due to the fact that $\tau=0$ is the end point of the Chebyshev interval, the Chebyshev polynomials are uniformly bounded there, and  one can show that 
 \begin{equation}
  \sum_{j=1}^{n+1} \abs{\gamma_j} = \cO(\log n). 
 \end{equation}
\end{remark}

\section{First-order Approximation by Hamiltonian Dilation} 

Another type of algorithm is based on a dilation with 
a dilated Hamiltonian \cite{cleve2016efficient,Ding2024}.  The main idea is to simulate the non-unitary Lindblad dynamics by embedding them into a unitary evolution on a larger Hilbert space. To this end, one introduces an ancilla register with Hilbert space  
\begin{equation}
\mathcal{H}_A = \mathrm{span}\{\ket{0}, \ket{1}, \ket{2}, \dots, \ket{J}\},
\end{equation}
which acts as a control register. A unitary channel is then constructed on the composite space \(\mathcal{H}_A \otimes \mathcal{H}_S\), where \(\mathcal{H}_S\) is the system Hilbert space. This unitary approximates the dissipative evolution of the open system when followed by a partial trace over the ancilla register. This operation is repeated at every time step. Specifically, let \(\epsilon = \sqrt{\tau}\), and define the block-structured Hamiltonian as follows:
\begin{equation}\label{eq:dilated-H}
  H := \epsilon^{2} H_0 + \epsilon H_1,
\end{equation}
with
\begin{equation}\label{eq:H0H1-def}
\begin{aligned}
  H_0 &= \ketbra{0}{0} \otimes H_S, \\
  H_1 &= \sum_{j=1}^{J} \left( \ketbra{j}{0} \otimes L_j + \ketbra{0}{j} \otimes L_j^{\dagger} \right).
\end{aligned}
\end{equation}

Then, one step of the Lindblad evolution can be approximated by
\begin{equation}\label{dilatedH}
  e^{\tau \mathcal{L}} \rho   \approx   \mathcal{K}(\tau)\rho 
    :=   \tr_A \left( e^{-iH} \left( \ketbra{0} \otimes \rho \right) e^{iH} \right).
\end{equation}
This operation is repeated at each time step. By \cref{rhorhohat}, for longer simulation times, the Lindblad operator norm $\ell$ is multiplied by $T$. Thus, we can assume $\ell >= J+1$ in this section.

Unlike the Kraus-form approximation in \cref{eqs:Kraus_operators}, this Stinespring-type dilation does not require amplitude amplification. On the other hand, the associated error expansion \Cref{eq:kappa_expansion}, which is needed for the backward error analysis, is more involved. In \Cref{diH-expansion}, we carry out this analysis, and our main result is summarized below.

\begin{theorem}[Local Error Expansion]
    The dilated Hamiltonian approach \eqref{dilatedH} produces an approximation as a reduced density operator $\rho_\text{R}$. It can be expanded as,
\begin{equation}
   \rho_\text{R} =  \sum_{k\geq 0}\ \tau^{k}\rho^{(2k)}_\text{R}.
\end{equation}
In particular,  $\rho^{(0)}_\text{R}= \rho$ and $\rho^{(2)}_\text{R}= \cL \rho$.
Furthermore for every integer \(k\ge2 \)
\begin{equation}\label{eq:even-bound}
    \norm{\rho^{(2k)}_\text{R}}
        \le  (J+1) \frac{\ell^{  k}}{k!}.
\end{equation}
\end{theorem}
This shows that the corresponding terms in expansion \eqref{eq:kappa_expansion} have the following bounds,
\begin{equation}
\cM_k \le  (J+1) \frac{\ell^{  k}}{k!}, \quad k\ge 2.
\end{equation}

Building on the derivative bounds above, we obtain a concentration bound on the sampling complexity, analogous to the previous case. The result is formalized in the following theorem.

\begin{theorem}\label{main-thm-dH}
Given $\varepsilon>0$. Assume that the simulation time satisfies \cref{T>T*}. 
     The Richardson extrapolation $p_n(\tau)$ of \(f(\tau)\) using $n+1$  perturbed Chebyshev nodes, with each expectation value $f(\tau_j) $, generated from the dilated Hamiltonian approximation \eqref{dilatedH} and  sampled with $N_S$ shots, is guaranteed to produce an estimator 
     \begin{equation}
         \mathbb P\left( \abs{p_n(0) - f(0) } < \varepsilon \right) > 1 - \delta,
     \end{equation}
     provided that $n = \Omega(\log \frac{1}{\varepsilon})$,  and $N_S= \Omega \left( \frac{1}{\varepsilon^2}\log \frac{1}{\delta} \right). $
     In addition, to prepare $\rho(T)$ (i.e., $ {\rho}(1) $ from \eqref{rhorhohat}), the maximum circuit depth of is $d_{\max} = \frac{1}{\tau_1}= \frac{T^2}{\tau_{\max}} n^2 = (\ell T)^2 (\log (\ell T)) \log^2 \frac{1}{\varepsilon}. $
\end{theorem} 

\section{Numerical Experiments}
\label{sec:numerics}

\begin{figure*}[htbp]
  \centering
  \includegraphics[scale=0.22]{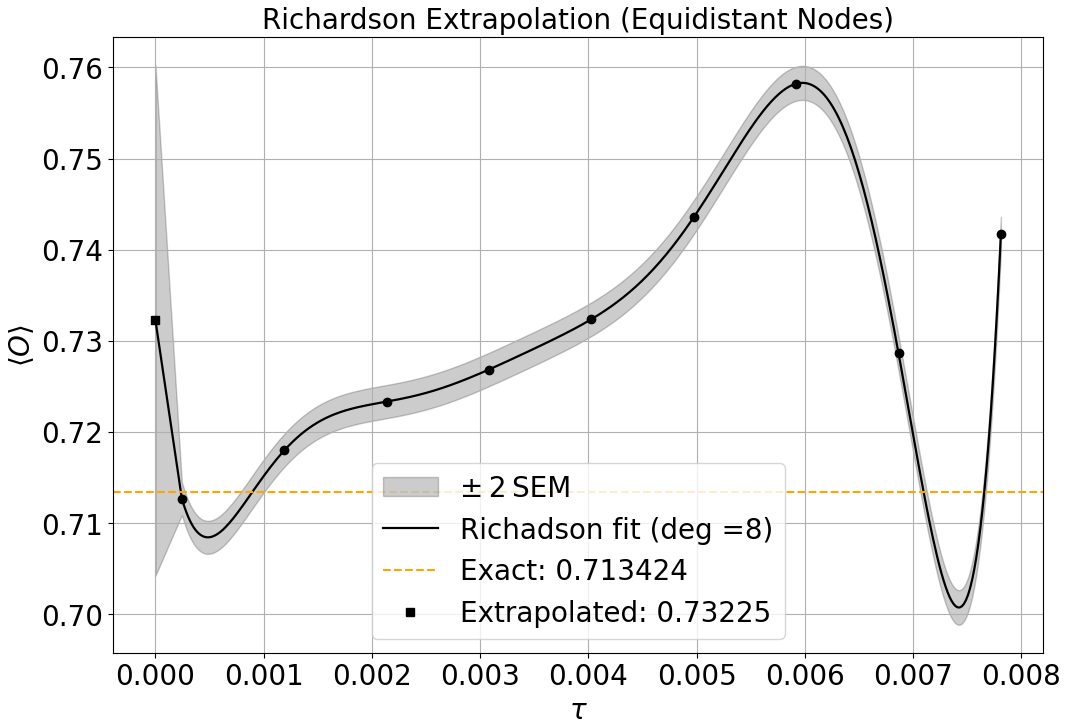}
  \includegraphics[scale=0.22]{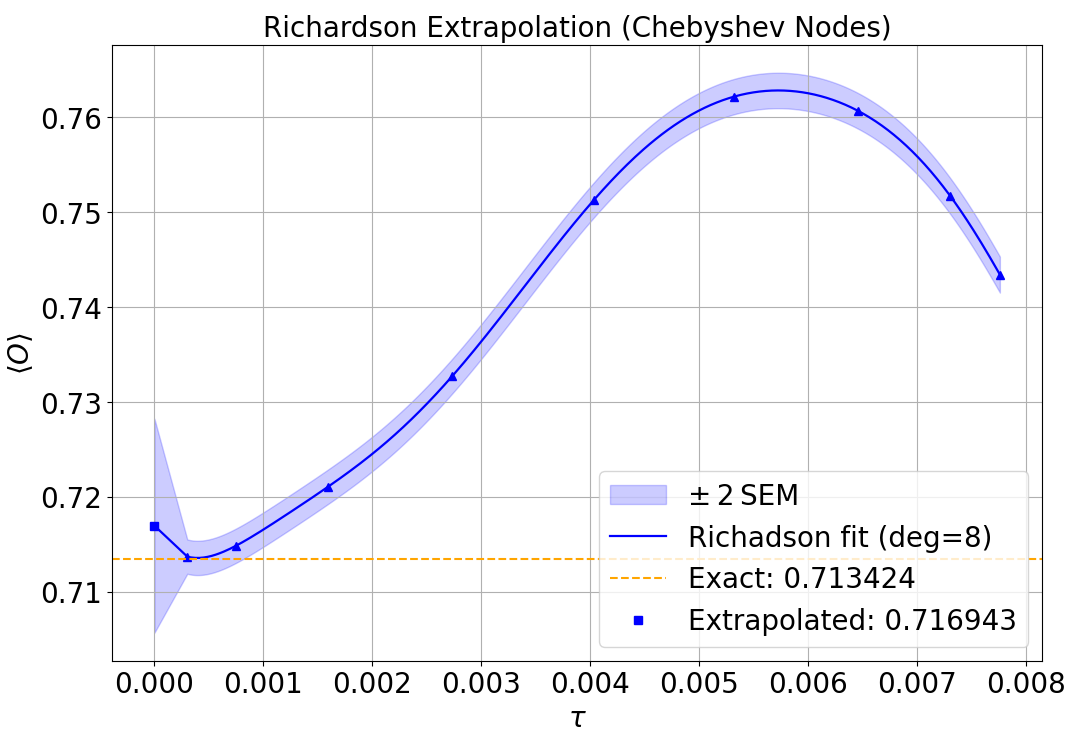}
  \caption{
  {Richardson extrapolation of $\langle O \rangle(\tau)$ using first-order Kraus evolution. The curves represent a degree-8 polynomial interpolating noisy data at $n=9$ points and then extrapolating to $\tau = 0$. $N_\mathrm{shots}=2 \times 10^7$.
  \textbf{Left}: Equidistant time steps. \textbf{Right}: perturbed Chebyshev time steps.  
}
}
  \label{fig:kraus-richardson}
\end{figure*}

\begin{figure*}[htbp]
  \centering
  \includegraphics[scale=0.22]{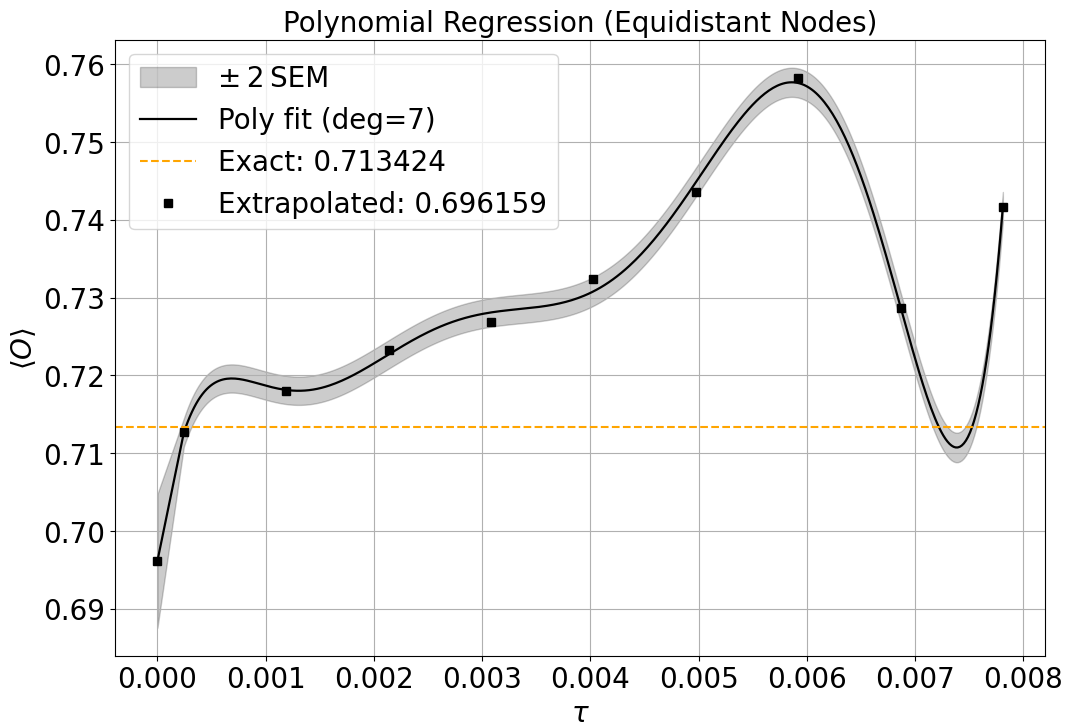}
  \includegraphics[scale=0.22]{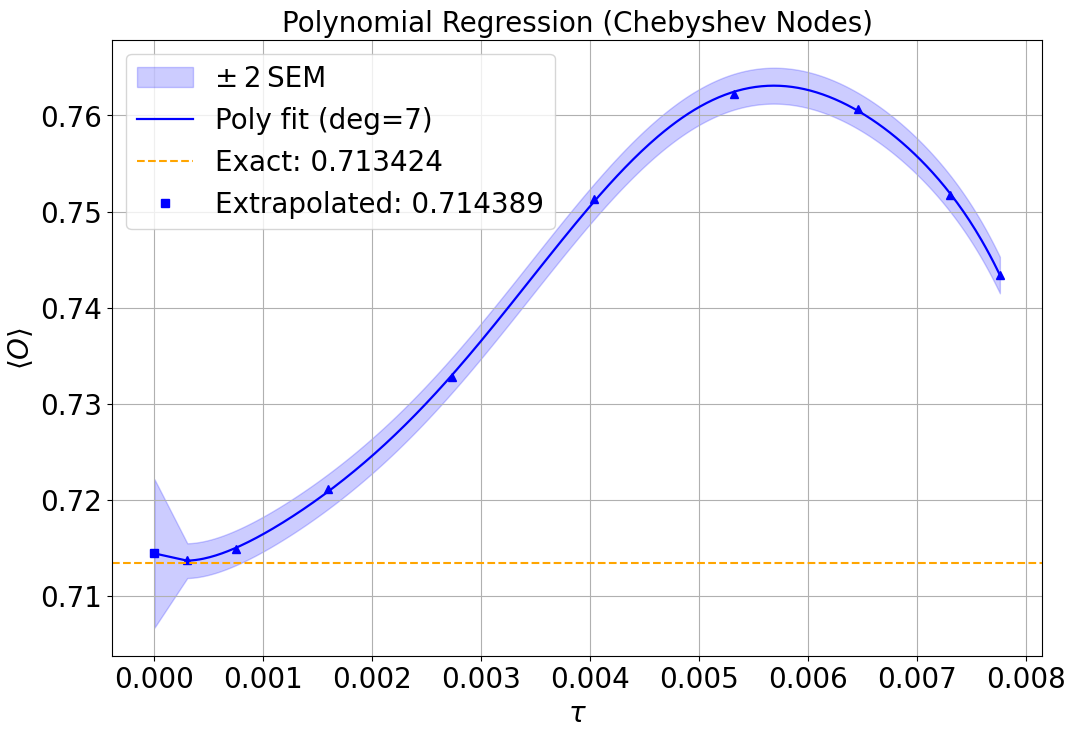}
  \caption{{Least-squares based extrapolation of $\langle O \rangle(\tau)$ with data generated by first-order Kraus-form approximation. $N_\mathrm{shots}=2 \times 10^7$.
  \textbf{Left}: Equidistant time steps. \textbf{Right}: Chebyshev time steps. The continuous curves represent degree-7 polynomials fit to $9$ noisy data points, then extrapolated to $\tau = 0$. }
}
  \label{fig:kraus-poly}
\end{figure*}

We now evaluate the extrapolation strategies discussed in previous sections. We consider two specific extrapolation methods,
\begin{enumerate}
    \item \emph{Richardson extrapolation}, which corresponds to a polynomial interpolation on $n+1$ data points. The interpolating polynomial $p_n(\tau)$ has degree at most $n$ and is evaluated at $\tau = 0$.
    \item \emph{Polynomial regression}, which seeks a least-squares fit to the $n+1$ data points using a polynomial $p_m(\tau)$ of degree $m < n$.
\end{enumerate}

In addition to the extrapolation methods, we also choose the nodes $\{\tau_j\}$ from either equidistant nodes $[0, \tau_{\max}]$, with the first node being $\tau_{\min} = \tau_{\max}/(n+1)$. Perturbed Chebyshev nodes in the same interval are also implemented with perturbations ensuring that $T/\tau_j$ are moved to the closest integers.

\subsection{Kraus Form approximation}
To test the first-order approximation in the Kraus form \Cref{eqn:Kraus_form}, we consider Lindblad dynamics with a single jump operator~$L$. The operators $L$,  $H$, and the observable $O$ are generated randomly in $\mathbb{C}^{16\times 16}$. 
For comparison purposes, the exact extrapolated value is computed directly from the trace with the density operator $\rho(T)$ computed using a standard ODE solver with a very small step size.  

\Cref{fig:kraus-richardson} shows the extrapolation results using stepsizes at $9$ points  with equidistant and Chebyshev nodes, and $T=10.$  Even with a small number of nodes, the interpolating polynomial already begins to exhibit oscillations towards the end point of the interval $[0, \tau_{\max}]$, thus leading to large extrapolation error.

To reduce overfitting, we perform a least-squares approach using polynomials of degree~$m=7$, and show the results in \Cref{fig:kraus-poly}.
These results also show that extrapolation on the Chebyshev grid consistently outperforms equidistant schemes, both in terms of accuracy and robustness to sampling noise.

\begin{figure*}[htbp]  
  \centering
  \includegraphics[scale=0.21]{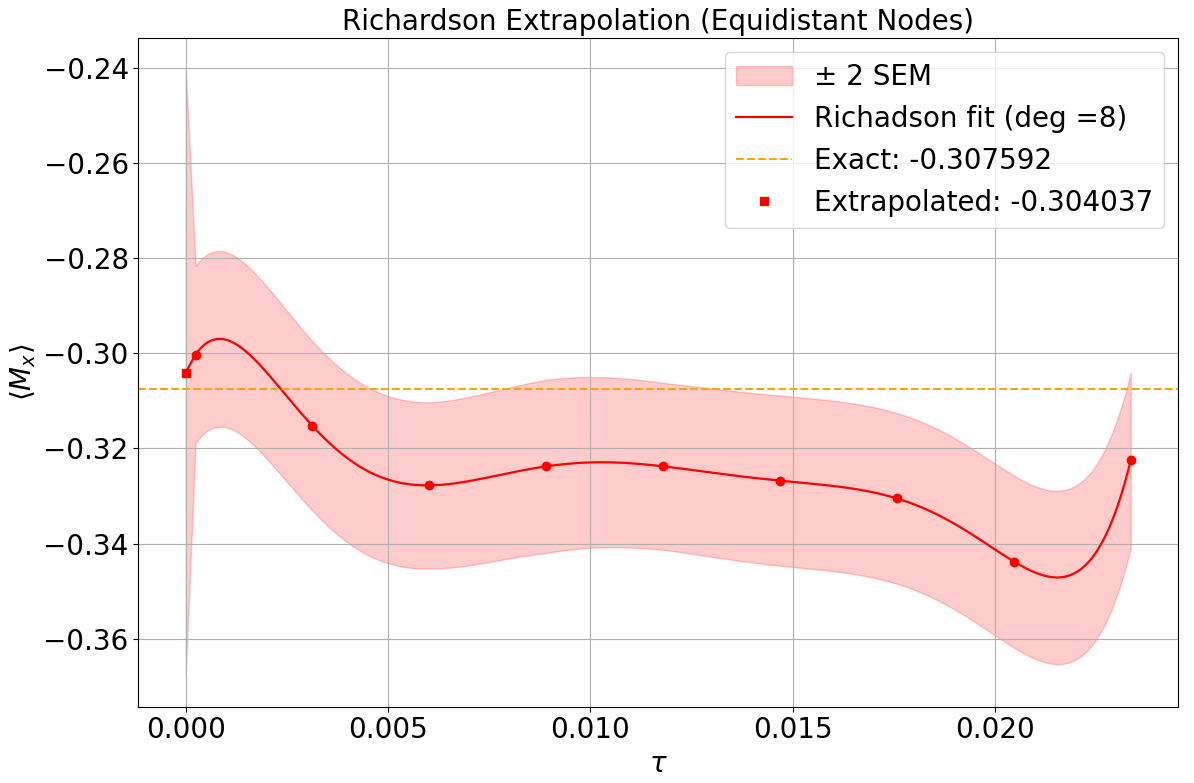}
  \includegraphics[scale=0.21]{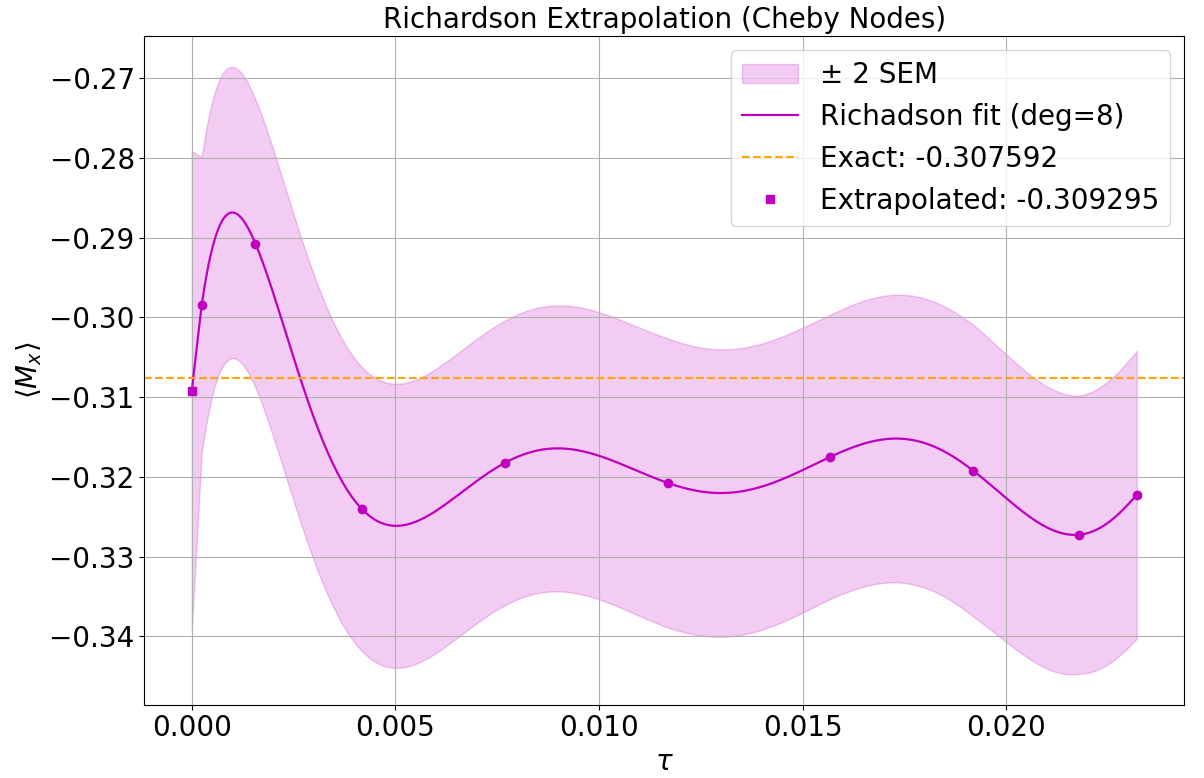}
  \caption{{Richardson extrapolation of $\langle M_x\rangle$ via Kraus-form approximation. $N_\mathrm{shots}=2 \times 10^3$. \textbf{Left}: Equidistant time steps. \textbf{Right}: Chebyshev time steps. Both are computed using a degree‑8 polynomial interpolating 9 noisy data points. Chebyshev nodes reduce bias and variance, enhancing agreement with the true expectation value.}
}
  \label{fig:kraus_tfim_richardson}
\end{figure*}

\begin{figure*}[htbp]  
  \centering
  \includegraphics[scale=0.21]{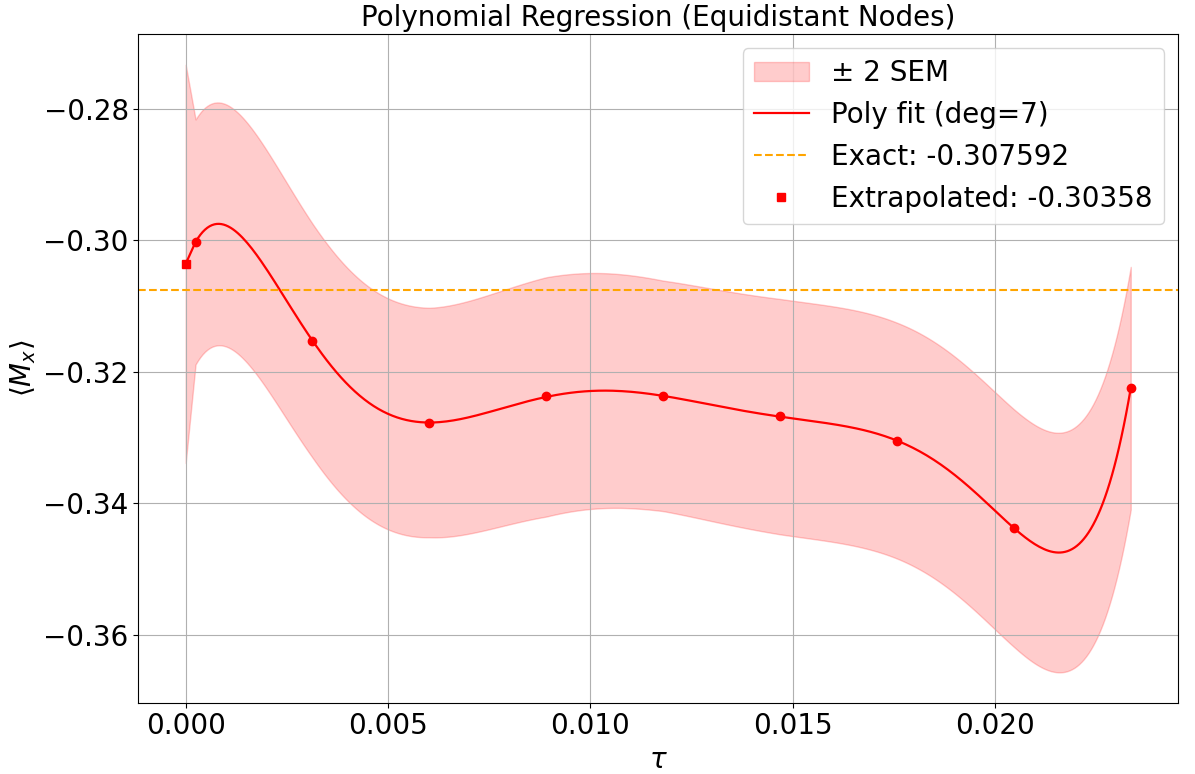}
  \includegraphics[scale=0.21]{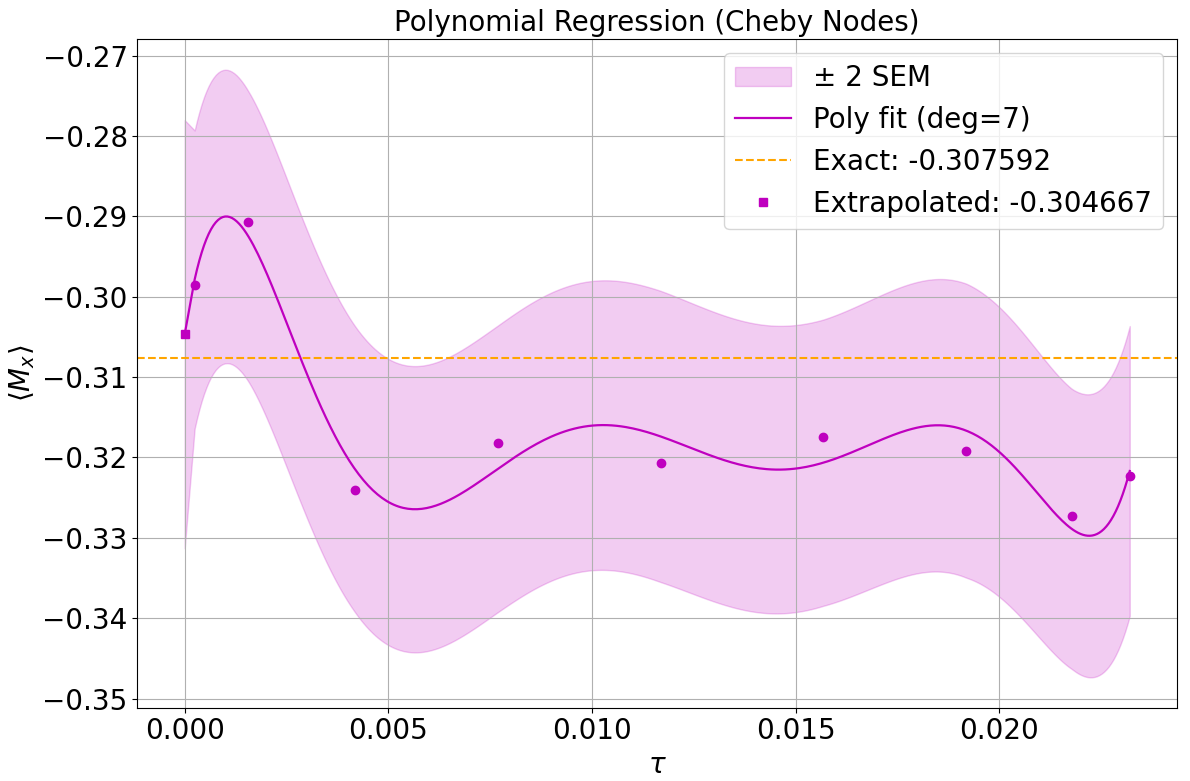}
  \caption{{Least-square based extrapolation of $\langle M_x\rangle$ via Kraus-form approximation. $N_\mathrm{shots}=2 \times 10^3$. \textbf{Left}: Equidistant time steps. \textbf{Right}: Chebyshev time steps. A degree‑7 polynomial fit is used on 9 noisy data points and is shown as a curve in each plot.}
}
  \label{fig:kraus_tfim_poly}
\end{figure*}

\noindent
To further assess the behavior of the first-order Kraus approximation under physically structured dynamics, we extend the same extrapolation framework to a four-qubit transverse-field Ising model (TFIM).  
We consider a \textit{transverse-field Ising model (TFIM)}, whose system Hamiltonian is defined as  
\begin{equation} \label{eq:tfim}
    H_S
    = \sum_{q=0}^{n_q-1} \left( \frac{\omega}{2}\,\sigma_z^{(q)} + \frac{\Omega_R}{2}\,\sigma_x^{(q)} \right)
      + J \sum_{q=0}^{n_q-2} \sigma_x^{(q)} \sigma_x^{(q+1)}.
\end{equation}
Here, $\omega$ denotes the local spin frequency, $\Omega_R$ the transverse driving strength, and $J$ the nearest-neighbor coupling constant.  
In this setting, both the system Hamiltonian and the jump operators in \Cref{eqn:Kraus_form} inherit the spin-chain structure of the TFIM. 
The open-system dynamics is modeled using \textit{local amplitude-damping channels} acting independently on each qubit, capturing relaxation processes at the level of individual spins.
Each jump operator corresponds to a specific decay channel.
$$V_q = \sqrt{\gamma} \sigma_-^{(q)},$$
where $\gamma$ is the decay rate and $\sigma_-^{(q)}$ is the lowering operator.  
We choose $n_q=4$, and parameters $\omega=1.0,  \Omega_R=0.8,  J = 0.3, \gamma=0.4$  in the numerical experiements.
 
The observable is chosen as the normalized total magnetization in the $x$-direction,
\begin{equation} \label{eq:magnetization}
M_x = \frac{\sum_{q=0}^{n_q-1} \sigma_x^{(q)}}{\left\| \sum_{q=0}^{n_q-1} \sigma_x^{(q)} \right\|}.
\end{equation}

\begin{figure*}[htbp]
  \centering
  \includegraphics[scale=0.23]{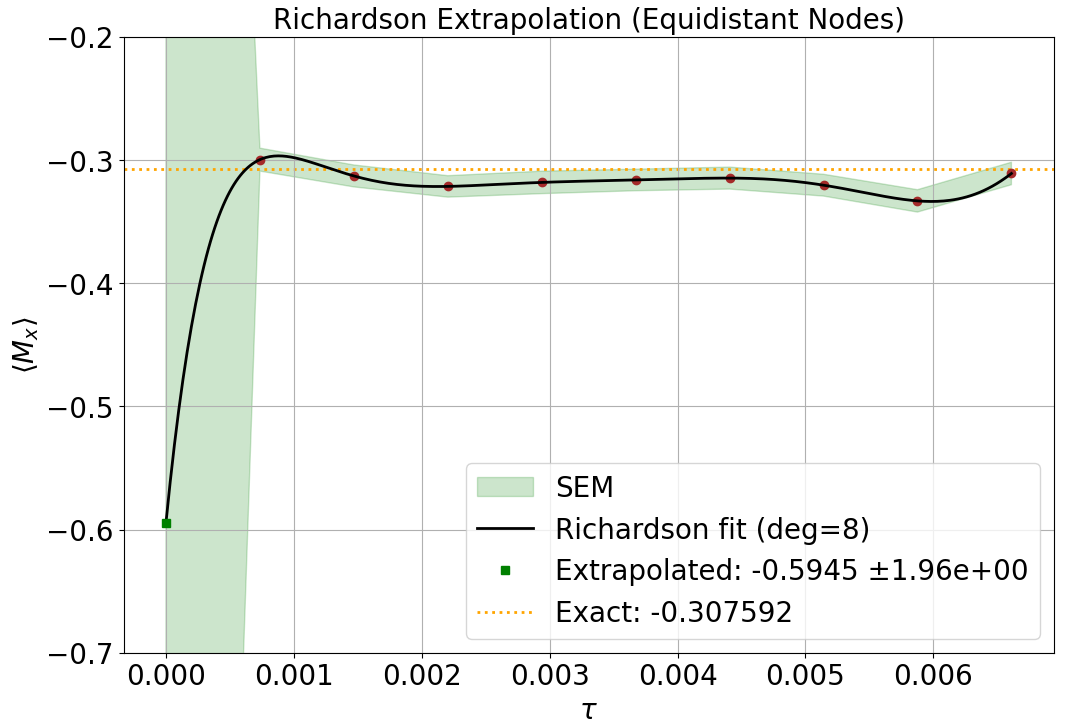}
  \includegraphics[scale=0.23]{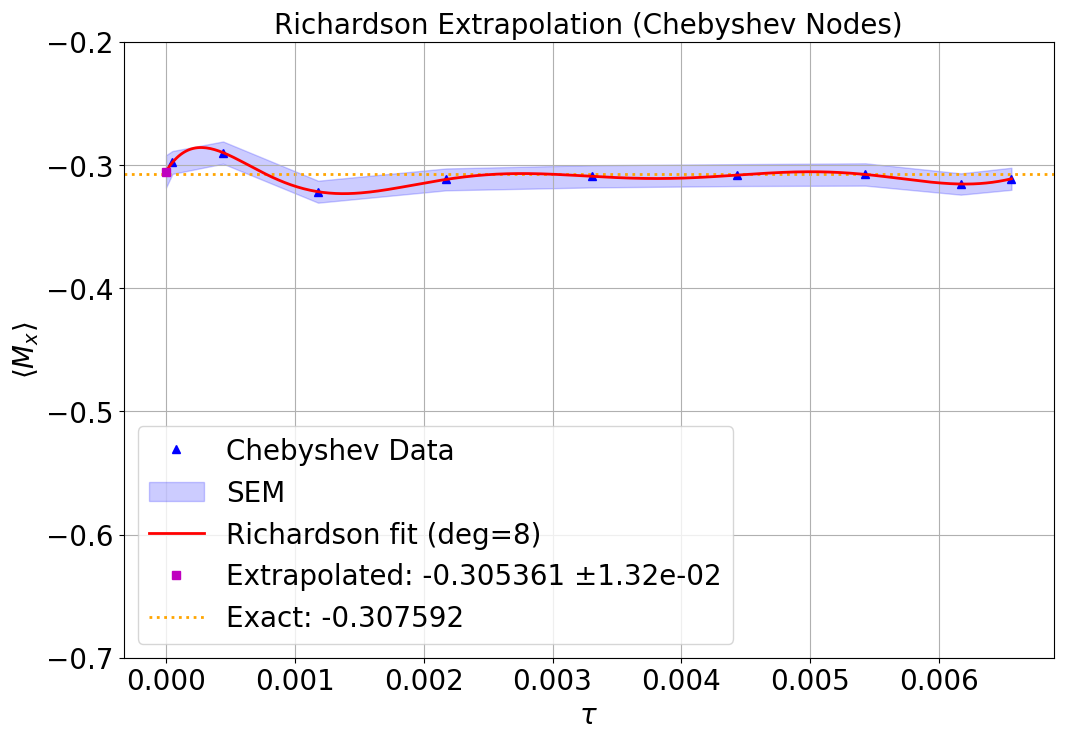}
  \caption{{Richardson extrapolation of $\langle M_x\rangle$ via dilated-Hamiltonian Lindblad evolution.  $N_\mathrm{shots}=2 \times 10^3$. 
\textbf{Left}: Equidistant time steps. \textbf{Right}: Chebyshev time steps. Both are computed using a degree‑8 polynomial interpolating 9 noisy data points. Chebyshev grids reduce bias and variance, enhancing agreement with the true expectation value.}
}
  \label{fig:dilated-richardson}
\end{figure*}

\begin{figure*}[htbp]  
  \centering
  \includegraphics[scale=0.23]{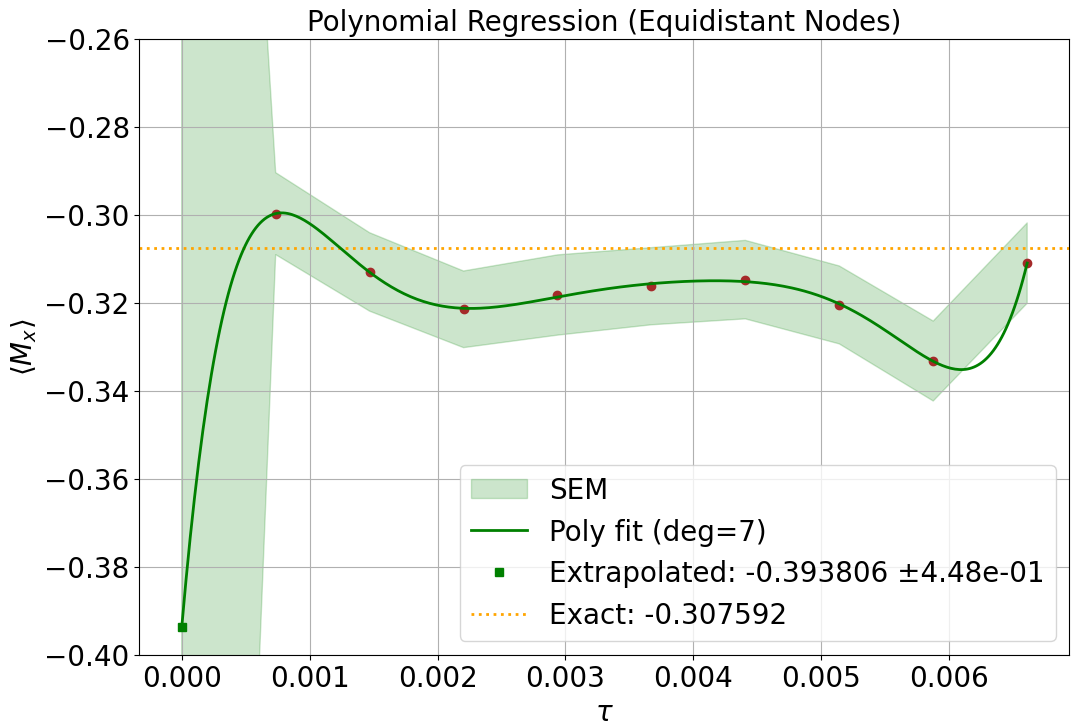}
  \includegraphics[scale=0.23]{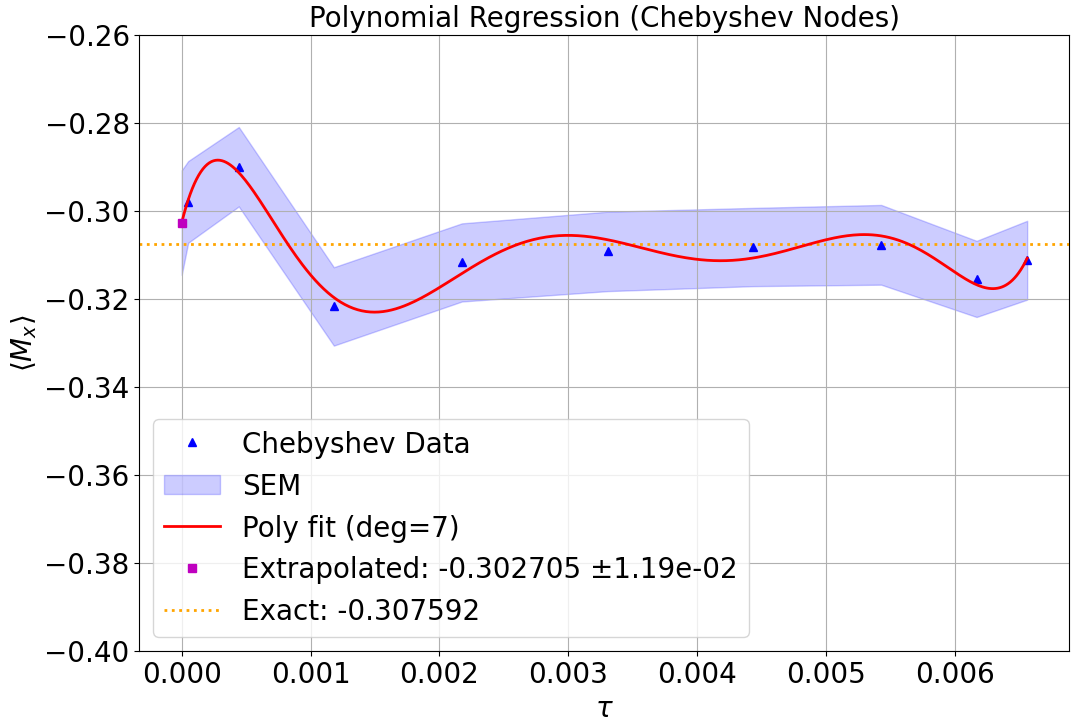}
  \caption{{Least-square based extrapolation of $\langle M_x\rangle$ via dilated-Hamiltonian Lindblad evolution. $N_\mathrm{shots}=2 \times 10^3$. \textbf{Left}: Equidistant time steps. \textbf{Right}: Chebyshev time steps. A degree‑7 polynomial fit is used on 9 noisy data points and is shown as a curve in each plot.}
}
  \label{fig:dilated-polynomial}
\end{figure*}

In this setting, the Kraus operators in the approximation in \Cref{eqn:Kraus_form} will depend on the system Hamiltonian and jump operators.  
Specifically, each time step evolves the density matrix through a first-order Kraus map \(\rho \;\mapsto\; F_0 \rho F_0^\dagger + \sum_j F_j \rho F_j^\dagger\),
\[
F_0 = I + \Delta t(-iH_S + G_0), \quad
F_j = \sqrt{\Delta t}\,V_j,
\]
with $G_0 = -\tfrac{1}{2}\sum_j V_j^\dagger V_j$.
As shown in \cref{fig:kraus_tfim_richardson}, and \cref{fig:kraus_tfim_poly}, the extrapolated values obtained from this TFIM Kraus evolution display the same qualitative trends as in the random-matrix experiment: both Richardson interpolation and polynomial regression on equidistant nodes are prone to oscillatory behavior, while the perturbed Chebyshev nodes yield smaller bias in the extrapolated expectation value $\langle M_x\rangle(\tau\!\to\!0)$.  
This provides additional evidence that node selection, rather than model-specific dynamics, primarily governs the numerical stability of polynomial extrapolation in the presence of stochastic sampling noise.

\subsection{Dilated Hamiltonian Approximations}

We now repeat the numerical experiments with approximations generated by the dilated Hamiltonian method, followed by extrapolations to $\tau=0.$ We consider the same TIFM model in \eqref{eq:tfim}.  Similarly, for the observable, we compute the normalized total magnetization in the x-direction in \cref{eq:magnetization}.
As shown in \Cref{fig:dilated-richardson,fig:dilated-polynomial}, a direct extrapolation using equidistant nodes typically leads to large statistical error, while the extrapolation based on the Chebyshev nodes provides a more robust estimate for the extrapolated values.

\section{Summary and further discussions}
\label{sec:conclusion}

We have carried out a comprehensive bias–variance analysis of Richardson-style
\emph{algorithmic error mitigation} for first-order quantum algorithms that simulate
Lindblad dynamics.  Without invoking any \emph{a-priori} smoothness assumptions, we proved that an
$n=\Omega \bigl(\log(1/\varepsilon)\bigr)$-point extrapolator
reduces the  maximum circuit depth required to reach precision~$\varepsilon$
from $\mathcal{O}((\ell T)^{2}/\varepsilon)$ to
$\mathcal{O} \bigl((\ell T)^{2}\log^2(1/\varepsilon)\bigr)$, while preserving the
standard $1/\varepsilon^{2}$ sampling complexity.
Extensive numerical experiments confirmed that the predicted depth savings
translate into measurable fidelity gains on noisy hardware. This framework opens the door to provably robust quantum simulations of dissipative dynamics on NISQ devices,  an essential ingredient for many problems in quantum chemistry and physics.

\paragraph{Extensions.}
\begin{itemize}
  \item \textbf{Higher-order integrators.}  The backward-error machinery developed
        here applies verbatim to second-order product formulas
        \cite{Childs2021Trotter} and quantum channel or dilation schemes
        \cite{li2022simulating,Ding2024}.  Updating the algebraic expansion
        yields Richardson weights of the same closed form; the asymptotic depth
        further improves by a constant factor.
    \item \textbf{Integration with Trotter Extrapolation.}  The dilated Hamiltonian approach is implemented by Hamiltonian simulation, for which a Trotter algorithm provides a simple implementation. More importantly, we can integrate the current approach with the extrapolation methods for 
    Trotter algorithms \cite{Low2019,watson2024exponentially}. Together, they can provide a powerful near-term approach for simulating open quantum systems. 
  \item \textbf{Joint mitigation of algorithmic and physical errors.}
        Recent work \cite{Endo2019,mohammadipour2025direct} shows that step-size
        extrapolation can be combined with zero-noise extrapolation.
        Integrating our bias bounds with a hardware-noise model is a promising
        route to end-to-end error budgets on NISQ devices.
\end{itemize}

\paragraph{Outlook.}
The analysis confirms that algorithmic extrapolation is not merely a heuristic
but a provably effective depth-reduction strategy for open-system simulation.
Because the technique requires no additional quantum resources beyond repeated
execution of shallow circuits, it can be implemented on current hardware
without modification.  Future work includes extending the framework to
time-dependent Lindblad generators, Lindblad control problems, 
 as well as exploring adaptive grid selection to
optimize the bias–variance trade-off in real time.
\begin{acknowledgments}
This research is supported by the NSF Grants No. DMS-2111221 and No. CCF-2312456.
\end{acknowledgments}

\bibliographystyle{apsrev4-2}
\bibliography{references_aps}  

\onecolumngrid
\section{Supplemental Material}
\begin{appendix}

\section{Detailed derivation of the error expansion}

\textbf{Lemma \ref{lemma:discretization_expansion}}. 
\textit{Let \(\rho_\tau (t)\) denote the approximate solution at time \(t\) with stepsize \(\tau\); there exists a sequence of smooth functions \(\Gamma_k(t)\) such that
\begin{equation} \label{eq:rho_expansion_proof}
    \rho_\tau(t) = \rho(t) + \tau \Gamma_1 (t) + \tau^2 \Gamma_2(t) + \dots ,
\end{equation}
where \(\rho(t)\) is the solution to the exact evolution \eqref{eq:lindblad}. For a discrete solution  \eqref{eq:kappa_expansion}, these coefficient matrices satisfy the initial value problems \(\Gamma_k(0) = 0\) for all \(k \geq 1\), and evolution equations,
\begin{equation} \label{eq:gamma_ode_proof}
    \Gamma'_{k-1}(t) = \mathcal{L} \Gamma_{k-1}(t)  - \frac{\mathcal{L}^k \rho (t)}{k!} + \sum_{i=2}^k \left( \mathcal{M}_i \Gamma_{k-i}(t) -  \frac{\Gamma^{(i)}_{k-i} (t)}{i!} \right).
\end{equation}
}

\begin{proof}
\label{proof:coefficient_first_order}
    The first part of the lemma follows Theorem 4.37 from \cite{deuflhard2012scientific}. The part that remains to prove is \eqref{eq:gamma_ode_proof}. Note that we are using a numerical scheme that updates the approximate solution in each step according to:
    \begin{equation}\label{eq:numerical_step}
        \rho_{n+1} = \mathcal{K}(\tau) \rho_n.
    \end{equation}
     On the left-hand side of \eqref{eq:numerical_step}, we can expand \(\rho_{\tau}(t+\tau)\). Since \( \rho(t+\tau) = e^{\tau \mathcal{L}} \rho(t)\),
we have
\begin{equation}
    \rho_{\tau}(t+\tau) = e^{\tau\mathcal{L}}\rho(t) + \tau\Gamma_1(t+\tau) + \tau^2 \Gamma_2(t+\tau) + \cdots.
\end{equation}
By expanding \(e^{\tau \mathcal{L}}\) into its Taylor series, \(\displaystyle e^{\tau\mathcal{L}} = \sum_{k=0}^{\infty} \frac{\tau^k}{k!} \mathcal{L}^k\), we find that
\begin{equation}\label{eq:left_hand_side}
    \rho_{\tau}(t+\tau) = \rho(t) + \tau \left( \mathcal{L} \rho(t) + \Gamma_1(t+\tau) \right) + \tau^2 \left( \frac{\mathcal{L}^2 \rho(t)}{2!} + \Gamma_2(t+\tau) \right) + \tau^3 \left( \frac{\mathcal{L}^3 \rho(t)}{3!} + \Gamma_3(t+\tau) \right) + \cdots
\end{equation}
which can be compactly written as
\begin{equation}
    \rho_{\tau}(t+\tau) = \rho(t) + \sum_{k=1}^{\infty} \tau^k \left( \frac{\mathcal{L}^k \rho(t)}{k!} + \Gamma_k(t+\tau) \right).
\end{equation}
Using the Taylor series, we can also expand \(\Gamma_k(t+\tau)\), obtaining
\begin{equation} \label{eq:Gamma_expanded}
    \Gamma_k(t+\tau) = \sum_{m=0}^{\infty} \frac{\tau^m}{m!}\,\Gamma^{(m)}_k(t).
\end{equation}
Substituting this expansion into \eqref{eq:left_hand_side}, we find
\begin{equation} \label{eq:modified_left_hand_side}
    \rho_{\tau}(t+\tau) = \rho(t) + \sum_{k=1}^{\infty} \tau^k \left( \frac{\mathcal{L}^k \rho(t)}{k!} + \sum_{j=0}^{k} \frac{\Gamma^{(k-j)}_j(t)}{(k-j)!} \right)
    = \rho(t) + \sum_{j=1}^{\infty} \tau^j \left( \frac{\mathcal{L}^j \rho(t)}{j!} + \sum_{i=0}^{j} \frac{\Gamma^{(i)}_{j-i}(t)}{i!}\right).
\end{equation}

On the right-hand side of \eqref{eq:numerical_step}, applying \(\mathcal{K}(\tau)\) from \eqref{eq:kappa_expansion} to the expansion
\begin{equation}
\rho_{\tau}(t) = \sum_{p=0}^{\infty} \tau^p \Gamma_p(t) ,
\end{equation}
yields
\begin{equation} \label{eq:modified_right_hand_side}
\rho_{\tau}(t+\tau)
= \sum_{\ell=0}^{\infty} \tau^\ell \mathcal{M}_\ell \left( \sum_{p=0}^{\infty} \tau^p \Gamma_p(t) \right)
= \rho(t) + \sum_{j=1}^{\infty} \tau^j \left( \sum_{p=0}^{j} \mathcal{M}_{j-p} \Gamma_p(t) \right),
\end{equation}
The equality of \eqref{eq:modified_left_hand_side} and \eqref{eq:modified_right_hand_side} implies that for every \(j \geq 1\):
\begin{equation}
\frac{\mathcal L^j\rho(t)}{j!}
+\sum_{i=0}^j \frac{\Gamma_{\,j-i}^{(i)}(t)}{i!}
  =  
\sum_{p=0}^j \mathcal M_{\,j-p}\,\Gamma_p(t).    
\end{equation}
Rewriting the left sum by isolating the \(i=0\) term \(\Gamma_j(t)\)
yields the discrete‐time relation
\begin{equation}
    \Gamma_j(t)
=\sum_{p=0}^j \mathcal M_{\,j-p}\,\Gamma_p(t)
-\sum_{i=1}^j\frac{\Gamma_{\,j-i}^{(i)}(t)}{i!}
-\frac{\mathcal L^j\rho(t)}{j!} 
\end{equation}
Rearranging the terms gives,
\begin{equation} 
    \Gamma'_{k-1}(t) = \mathcal{L} \Gamma_{k-1}(t)  - \frac{\mathcal{L}^k \rho (t)}{k!} + \sum_{i=2}^k \left( \mathcal{M}_i \Gamma_{k-i}(t) -  \frac{\Gamma^{(i)}_{k-i} (t)}{i!} \right). 
\end{equation}
$\qed$
\end{proof}

\section{Bounding the coefficients \texorpdfstring{$c_{i,j,k}$}{c(i,j,k)}}\label{app:c_bound}

\textbf{Lemma \ref{lemma:c_bound_kraus}} .
    Let \(c_{i,j,k} \geq 0\) be as in \textit{Definition} \ref{def:c_s_kraus}. Then for all \(i,j,k\ge0\),
    \begin{equation} \label{ineq:c_bound_proof}
        c_{i,j,k}   \le   \frac{C^{i+k}_1 \, C^{k}_2}{j!}.
    \end{equation}
    where 
      \begin{equation}\label{def:C1_proof}
      C_1 := \max \{ B, \ell (e+1), 1 \}, \text{ and  } \, C_2 \geq (e+1) \log (C_1)
  \end{equation}

\begin{proof} \label{proof:c_bound_kraus}
    In \textit{Definition} \ref{def:c_s_kraus}, \(j \leq k\). Thus, we need to do induction on \(k\) for different values of \(i\):
    \\
    \textit{Base Case:} For \(k = 0\), and any \(i, j \geq 0\),
    \begin{equation}
        c_{i,j,0} = \delta_{j,0} \ell^i \leq \frac{C^{i+0}_1 \, C^0_2}{j!} = \frac{C^i_1}{j!} .
    \end{equation}
    \textit{Inductive Step:} Suppose \eqref{ineq:c_bound_proof} holds for all \(k' < k\). Then we have the following cases.
    \begin{itemize}
        \item \textit{Case \(i = 0\)}: For \(j=1, 2, \cdots, k\), using \textit{Definition} \ref{def:c_s_kraus}, and the induction hypothesis,
        \begin{align}
            c_{0,j,k} & =  \, \frac{B}{j}  \cdot c_{0,j-1,k-1} + \delta_{j,1} \cdot \frac{\ell^{k+1}}{(k+1)!} + \sum_{p=1}^{k -j} \frac{c_{p+1,j-1,k-p}}{j(p+1)!} \ \\
            & \leq \frac{B}{j} \cdot \frac{C^{k-1}_1 \, C^{k-1}_2}{(j-1)!} +  \frac{\ell^{k+1}}{(k+1)!} + \sum_{p=1}^{k-j} \frac{C_1^{k+1} \, C_2^{k-p}}{j! (p+1)!} \\
            & \leq \frac{C_1^{k} \, C_2^{k}}{j!} \left( \frac{B}{C_1 \, C_2} + \frac{\ell^{k+1}}{C^k_1 \, C^k_2} + \sum_{p=1}^{k-j} \frac{C_1}{C^p_2 \, (p+1)!}\right) \\
            & \leq \frac{C^k_1 \, C^k_2}{j!} \left( \frac{1}{e+1} + \frac{1}{e+1} + \frac{1}{e+1} \right) \leq  \frac{C^k_1 \, C^k_2}{j!} .
        \end{align}

         We note that because of the assumption \(\ell \ge 1\), we have \(C_1 \geq e+1\).To derive the last line notice that \(\frac{B}{C_1} \leq 1\), and \(\frac{1}{C_2} \leq \frac{1}{e+1}\). Similarly \(C^k_1 \, C^k_2 \geq \ell^{2k} (e+1)^{3k}\), then \(\frac{\ell^{k+1}}{C^k_1 \, C^k_2} \leq \frac{1}{e+1}\). Also,  \(\displaystyle \sum_{p=1}^{k-j} \frac{C_1}{C_2^p(p+1!} \leq \sum_{p=1}^{\infty} \frac{C_1}{C_2^p (p+1)!} =  C_1 (e^{-C_2} - 1) \le \frac{C_1}{e^{C_2}} \le \frac{1}{e+1}\). 

        \item \textit{Case \(i \geq 1\)}: For \(j=0,1,\cdots, k-1\), using \textit{Definition} \ref{def:c_s_kraus}, and the induction hypothesis,
    \begin{align}
        & c_{i,j,k} =   \, \ell \cdot c_{i-1,j,k} +  B \cdot c_{i-1,j,k-1} + \delta_{j,0} \cdot \frac{\ell^{i+k}}{(k+1)!} + \sum_{p=1}^{k-j} \frac{c_{i+p,j,k-p}}{(p+1)!}  \\
        & \leq \ell \cdot \frac{C_1^{i+k-1} \, C_2^k}{j!} + B \cdot \frac{C_1^{i+k-2} \, C_2^{k-1}}{j!} +  \frac{\ell^{i+k}}{(k+1)!} + \sum_{p=1}^{k-j} \frac{C_1^{i+k} \, C_2^{k-p}}{j! \, (p+1)!} \\
        & \leq \frac{C_1^{i+k} \, C_2^k}{j!} \left(\frac{\ell}{C_1} + \frac{B}{C_1^2 \, C_2} + \frac{\ell^{i+k}}{C_1^{i+k} \, C_2^{k} \, (k-j+1)!} + \sum_{p=1}^{k-j} \frac{1}{C^p_2 \, (p+1)!} \right) \\
        & \leq \frac{C_1^{i+k} \, C_2^{k}}{j!} \left(\frac{1}{e+1} + \frac{1}{e+1} + \frac{1}{e+1} + \frac{e-2}{e+1} \right) \leq \frac{C_1^{i+k} \, C_2^k}{j!}.
    \end{align}
    Where we derived the last line similar to the previous case.
    \\
    For \( j = k \), similarly:
    \begin{align}
        c_{i,k,k} & = \ell \cdot c_{i-1,k,k}  \\
        & \leq \ell \cdot \frac{C_1^{i+k-1} \, C_2^k}{k!} \\
        & = \frac{C^{i+k} \, C_2^k}{k!} \cdot \frac{\ell}{C_1} \leq \frac{C^{i+k} \, C_2^k}{k!}.
    \end{align}
    \end{itemize}
    This completes the proof.
\end{proof}

\section{Bounding the derivatives of error expansions} \label{app:gamma_bound}

\textbf{Lemma \ref{lemma:gamma_bound}}. 
\textit{The coefficients in the expansion  of the density operator \(\rho_{\tau}\) in \Cref{lemma:discretization_expansion}, approximated by the Kraus operator \eqref{eqn:Kraus_form}, satisfy the following bound: 
\begin{equation}
    \bigl\|\Gamma_{k}^{(i)}(t)\bigr\|
    \le  
  P_{i,k}(t)
    =   
  \sum_{j=0}^k
  c_{i,j,k}\,t^{j} \, \text{ for } i\geq 0, k \geq 1, \quad  c_{0,0,k} = 0 , \quad c_{i,j,0}= \delta_{j,0} \ell^i,
\end{equation}
where $\Gamma_k^{(0)} = \Gamma_k$, $\Gamma_k^{(1)} = \Gamma'_k$, and the coefficients \(c_{i,j,k}\) are defined by the generating  sequence in \textit{Definition} \ref{def:c_s_kraus}.
}

\begin{proof} 
Given the Kraus operators
\(
F_0 = I + \left(-iH - \frac{1}{2}\sum_j L_j^\dagger L_j\right)\tau, \quad F_j = L_j\sqrt{\tau},
\) and 
the discrete evolution operator 
\(
\mathcal{K}(\tau)[\rho] = F_0 \rho F_0^\dagger + F_1 \rho F_1^\dagger + \cdots.
\)
By expanding \(\mathcal{K}\) in \(\tau\), we obtain:
\begin{equation}
\mathcal{K}(\tau)[\rho] = \rho + \tau \mathcal{L}[\rho] + \tau^2 A \rho A^\dagger + \mathcal{O}(\tau^2),
\end{equation}
where  \(
\mathcal{L}
\) is the Lindblad operator, and
\(A = -iH - \frac{1}{2}\sum_j L_j^\dagger L_j\).
Thus,  in the expansion \eqref{eq:kappa_expansion}, \(\mathcal{M}_0 = I\), \(\mathcal{M}_1 = \mathcal{L}\), \(\mathcal{M}_2 \rho = A \rho A^\dagger\), and \(\mathcal{M}_i = 0\) for all \(i \geq 3\).
Then, by Lemma \ref{lemma:discretization_expansion}, each $\Gamma_k(t)$ satisfies a linear ODE:
    \begin{equation}  \label{eq:err-eq}
    \Gamma'_{k}(t) = \mathcal{L} \Gamma_{k}(t)  + \mathcal{M}_2 \Gamma_{k-1}(t) - \frac{\mathcal{L}^{k+1} \rho (t)}{(k+1)!} - \sum_{p=1}^{k} \frac{\Gamma^{(p+1)}_{k-p} (t)}{(p+1)!}, \quad \Gamma_k(0) = 0. 
    \end{equation}

Using the variation of constants formula on \eqref{eq:err-eq}, we solve explicitly:
\begin{equation}
    \Gamma_k(t) = \int_0^t e^{(t - s)\mathcal{L}} (\mathcal{M}_2 \Gamma_{k-1}(s) - \frac{\mathcal{L}^{k+1} \rho (s)}{(k+1)!} - \sum_{p=1}^{k} \frac{\Gamma^{(p+1)}_{k-p} (s)}{(p+1)!}) \,ds.
\end{equation}

By definition, \(\|\mathcal{M}_2\| = \|\mathcal{M}_2 \rho \| \leq \|A\|^2 \leq B \), where \(B := (\|H\| + \frac{1}{2} \|V\|^2)^2 \). Taking norms, and assuming the exponential $e^{(t-s)\mathcal{L}}$ is uniformly bounded by 1, results in
\begin{equation}\label{ineq:Gamma_k_bound}
    \|\Gamma_k(t)\| \leq  \int_0^t B \ \|\Gamma_{k-1} (s) \| + \frac{\ell^{k+1}}{(k+1)!} + \sum_{p=1}^{k} \frac{\|\Gamma^{(p+1)}_{k-p} (s)\|}{(p+1)!} \, ds.
\end{equation}

Moreover, from Lemma \ref{lemma:discretization_expansion}, we also obtain the \(i\)-th derivative of \(\Gamma_k\):
\begin{equation}\label{higher-deriv}
    \Gamma^{(i)}_k(t) = \mathcal{L} \Gamma^{(i-1)}_{k}(t)  + \mathcal{M}_2 \Gamma^{(i-1)}_{k-1}(t) - \frac{\mathcal{L}^{k+i} \rho (t)}{(k+1)!} - \sum_{p=1}^{k} \frac{\Gamma^{(i+p)}_{k-p} (t)}{(p+1)!}.
\end{equation}

Meanwhile for \(i > 0\), we no longer have $\Gamma_k^{(i)}(0)=0$. Therefore we will bound it directly using \cref{higher-deriv}
\begin{equation}\label{ineq:Gamma_k_i_bound}
 \|\Gamma^{(i)}_k(t)\| \leq \ell  \|\Gamma^{(i-1)}_k(t)\| + B   \|\Gamma^{(i-1)}_{k-1}(t)\| + \frac{\ell^{k+i}}{(k+1)!} + 
 \sum_{p=1}^k \frac{\|\Gamma^{(i+p)}_{k-p}(t) \|}{(p+1)!}.
\end{equation}

We claim the following polynomial bound for \(\Gamma_k\):
\begin{equation}\label{ineq:Pki}
    \bigl\|\Gamma_{k}^{(i)}(t)\bigr\|
    \le  
  P_{i,k}(t)
    =   
  \sum_{j=0}^k
  c_{i,j,k}\,t^{j} \, , \quad  c_{0,0,k} = 0 \text{ for } k \geq 1, \quad c_{i,j,0}= \delta_{j,0} \ell^i.
\end{equation}
where the constants \(c_{i, j,k} \geq 0\) are based on \textit{Definition} \ref{def:c_s_kraus}. 

\smallskip

\textit{Base case:} For \(i = k = 0\), \(\Gamma_0^{(0)} (t) = \rho (t)\), and \(\|\rho (t)\| \leq 1\), which holds in \eqref{ineq:Pki}.

\smallskip

\textit{Inductive step:}  Suppose \eqref{ineq:Pki} holds for all \(i' < i\), and \(k' < k\).  We consider two sub‐cases:

\begin{enumerate}[label=\emph{\arabic*)}]
  \item \textit{for $k\ge1$, and $i=0$ ($j \geq 1$):}  Using \eqref{ineq:Gamma_k_bound} and the induction hypothesis
\begin{equation}
    \|\Gamma_k(t)\|  \le  \int_0^t
       \Bigl(
         B\sum_{j=0}^{k-1}c_{0,j,k-1}s^j
           +  \frac{\ell^{k+1}}{(k+1)!}
           +  \sum_{j=0}^{k-1}\frac{\sum_{p=1}^{\,k-j}c_{p+1,j,k-p}}{(p+1)!}s^j
       \Bigr)\,ds,
\end{equation}
    
  we integrate term–by–term, and collect the coefficient of $s^j$.  On the right hand side, using \textit{Definition} \ref{def:c_s_kraus}, one obtains 
  \begin{equation}
    \sum_{j=0}^{k} c_{0,j,k} \ t^j
       =   
    \sum_{j=0}^{k} \left(\frac{B}{j}\,c_{0,j-1,k-1}
      +  
    \delta_{j,1}\,\frac{\ell^{k+1}}{(k+1)!}
      +  
    \sum_{p=1}^{k-j}\frac{c_{\,p+1,j-1,k-p}}{j\,(p+1)!}  \right)\, t^j,
    \end{equation}
  as claimed.

  \item \textit{for $k\ge1$, and $i\ge1$:}  Using \eqref{ineq:Gamma_k_i_bound} and the induction hypothesis
  \begin{equation}
       \|\Gamma_k^{(i)}(t)\|
      \le  
    \ell\sum_{j=0}^{k}c_{i-1,j,k}t^j
      +  
    B\sum_{j=0}^{k-1}c_{i-1,j,k-1}t^j
      +  
    \frac{\ell^{k+i}}{(k+1)!}
      +  
     \sum_{j=0}^{k-1}\frac{ \sum_{p=1}^{k-j} c_{\,p+i,j,k-p}}{(p+1)!}t^j,
  \end{equation}

  On the right hand side, using \textit{Definition} \ref{def:c_s_kraus}, one obtains
  \begin{equation}
    \sum_{j=0}^{k} c_{i,j,k} \ t^j
       =   
    \ell c_{i-1,k,k} \, t^k +\sum_{j=0}^{k-1}\left(\ell\,c_{i-1,j,k}
      +  
    B\,c_{i-1,j,k-1}
      +  
    \delta_{j,0}\,\frac{\ell^{k+i}}{(k+1)!}
      +  
    \sum_{p=1}^{\,k-j}\frac{c_{p+i,j,k-p}}{(p+1)!}  \right) t^j\,,
  \end{equation}
  as required. 

\end{enumerate}

This closes the induction and completes the proof.
\end{proof}

\textbf{Corollary  \ref{cor:gamma_bound_1}}. 
\textit{  For $t\leq 1$, the coefficients in the expansion  of the density operator \(\rho_{\tau}\) in \Cref{lemma:discretization_expansion}, approximated by the Kraus operator \eqref{eqn:Kraus_form}, satisfy the following bound for every \(i,k\in\mathbb N\),
  \begin{equation}
       \norm{\Gamma^{(i)}_k(t) }     \le  
     e\,C_1^{\,i+k}\,C_2^{\,k},
  \end{equation}
  where \(C_1 := \max \{ B, \ell (e+1), 1 \}\), and \(C_2 \geq C_1 (e+1)\).
}

\begin{proof}\label{proof:gamma_bound}
  Assume \(t\le 1\).
 
  \begin{equation}
      \bigl\|\Gamma^{(i)}_{k}(t)\bigr\|
       \le  
     \sum_{j=0}^{k} c_{i,j,k}\,t^{j}
       \le  
     C_1^{\,i+k}\,C_2^{\,k}
     \sum_{j=0}^{k}\frac{t^{j}}{j!}
       \le  
     e\,C_1^{\,i+k}\,C_2^{\,k}.
  \end{equation}

  The first inequality is exactly the definition of \(P_{i,k}(t)\) in
  Lemma~\ref{lemma:gamma_bound}.  For the second inequality, insert the bound
  from Lemma~\ref{lemma:c_bound_kraus},
  \(c_{i,j,k}\le C_1^{\,i+k}C_2^{\,k}/j!\).
  Finally, because \(0\le t\le1\) we have
  \(\sum_{j=0}^{k}t^{j}/j!\le\sum_{j=0}^{\infty}1/j!=e\),
  giving the upper bound.
\end{proof}

\section{Error bounds} \label{app:bias_bound}
\textbf{Theorem \ref{thm:Gevrey-f}}. \textit{Let $O$ be a bounded observable, and 
  \(
      f(\tau)  :=  \tr   \bigl(\rho_\tau(T)\,O\bigr),
     \, 0 \le T \le 1, 
  \)
  where $\rho_\tau(T)$ is approximated using the Kraus form \eqref{eqn:Kraus_form}. Then $f(\tau)$ belongs to the Gevrey class on  $[0,\tau_{\max}]$ and satisfies the following bound,
    \begin{equation} \label{ineq:fk-bound_proof}
        \abs{f^{(k)}(\tau)} \leq \sigma \, \nu^k \, k!, \quad \, \qquad \text{for } \tau \in [0,\tau_{\max}], \quad k \ge 1,
    \end{equation}
    where \(\sigma := 2e \|O\|\) and \(\nu := 2C_1 C_2\). If \(T \le 1\), then \cref{ineq:fk-bound_proof} holds for any \(\tau_{\max} \le 1/(2\nu)\). When \(T \ge 1\), the scaling in \cref{rhorhohat} applies and the bound remains valid with \(\nu := 2C_1 C_2\, T^2 \log (\ell T)\).}

\begin{proof}
For a fixed final time $t$, Lemma~\ref{lemma:discretization_expansion}
 provides the step-size expansion
\begin{equation}
    \rho_{\tau}(t) = \sum_{n=0}^{\infty} \tau^n \Gamma_n(t), \qquad \Gamma_{0}(t)=\rho(t).
\end{equation}
Taking the trace against a bounded observable $O$ (\,$\|O\|<\infty$\,)
yields the absolutely convergent series
\begin{equation} \label{eq:f_expansion}
  f(\tau)=\tr\bigl(\rho_\tau(T)O\bigr)
         =\sum_{n=0}^\infty a_n\,\tau^{n},
  \qquad
  a_n:=\tr \  \bigl(\Gamma_{n}(T)O\bigr),
  \quad 0\le \tau \le \tau_{\max}.
\end{equation}
For $i=0$ and $T\le1$, Corollary~\ref{cor:gamma_bound_1} yields
$\|\Gamma_{n}(T)\|\le e\,C_{1}^{\,n}\,C_{2}^{\,n}$.
Consequently,
\begin{equation} \label{ineq:an_bound}
  |a_n|  \le  e\,\|O\|\,C_{1}^{\,n}\,C_{2}^{\,n}
             =  \frac{\sigma}{2}\,(C_{1}C_{2})^{n},
  \qquad \sigma:= 2 e\|O\|.
\end{equation}
For $k\geq 1$ we differentiate \Cref{eq:f_expansion} termwise:
\begin{equation}
    f^{(k)}(\tau)
  =\sum_{n=k}^{\infty} n(n-1)\dotsm(n-k+1)\,a_n\,\tau^{\,n-k}.
\end{equation}
Using the combinatorial identity
$n(n-1)\dotsm(n-k+1)=k!\binom{n}{k}$ and
changing index $m:=n-k$,
\begin{equation} \label{eq:fk_expansion}
      f^{(k)}(\tau)
  =k!\,\sum_{m=0}^{\infty} \binom{m+k}{k}\,a_{m+k}\,\tau^{\,m}.
\end{equation}
By insertting \eqref{ineq:an_bound} into \eqref{eq:fk_expansion}, we have:
\begin{equation}
  |f^{(k)}(\tau)|
  \le
  k!\,\frac{\sigma}{2}\,
  (C_{1}C_{2})^{k}
  \sum_{m=0}^{\infty}\binom{m+k}{k}\,(C_{1}C_{2} \tau)^{m}.
\end{equation}
For every $m\ge0$, and \(k \geq 1\) the standard estimate
$\binom{m+k}{k}\le 2^{\,m+k}$ holds.  With this,
\begin{equation}
      \sum_{m=0}^{\infty}\binom{m+k}{k}\,(C_{1}C_{2} \tau)^{m}
  \le
  2^{\,k}\sum_{m=0}^{\infty}(2 C_{1}C_{2} \tau)^{m}
  =\frac{2^{\,k}}{1-2 C_{1}C_{2} \tau}.
\end{equation}
Since $0 \le \tau \le \tau_{\max}$,
we have $0\le 2 C_{1}C_{2} \tau  \le \frac{1}{2}$. Hence, the geometric series converges.
Thus,
\begin{equation}
    |f^{(k)}(\tau)|
  \le \sigma \,  k!  \, (2 C_1 C_2)^k = \sigma \, \nu^k \, k!  \, ,
\end{equation}

If $T \geq 1$, the operator bound $\ell$ and $B$ appearing in Definition~\ref{def:c_s_kraus} must be replaced by $\ell_{\mathrm{new}} = \ell T$ and $B_{\mathrm{new}} = B T^2$, respectively. Consequently, the constants in Lemma~\ref{lemma:c_bound_kraus} are modified to $\tilde{C}_1 := \max\{B T^2, \ell T (e+1), 1\}$ and $\tilde{C}_2 \geq (e+1) \log \tilde{C}_1$. Thus, the constants $C_1$ and $C_2$ will be scaled to $C_1 T^2$ and $C_2\log (\ell T) $ accordingly. Similarly, $\nu$ will be scaled to $\nu T^2 \log (\ell T)$, and the corresponding interval for \(\tau\) shrinks to $[0, \tau_{\max}/(T^2 \log (\ell T))].$

\end{proof}

\section{The Lebesgue constant for the perturbed Chebyshev nodes}

\textbf{Lemma} \ref{lem:integer-reciprocal} (Effect of Perturbed Chebyshev Nodes on Variance Bound)
Let \(n \ge 2\) and \(\tau > 0\), and define Chebyshev nodes on \([0, \tau]\) by
\begin{equation}
  \xi_j = \frac{\tau}{2} \left(1 - \cos\left(\frac{2j - 1}{2n + 2}\pi\right)\right),
  \qquad j = 1, \dots, n + 1.
\end{equation}
If the time parameter $\hat{T}>\pi^{2}\tau n^{2}$, define the perturbed nodes by
\[
  k_j:=\bigl\lceil  \hat{T}/\xi_j\bigr\rceil,\qquad
  \tau_j:=\frac{\hat{T}}{k_j},\qquad j=1,\dots,n+1.
\]
Then the following statements hold:
\begin{enumerate}
  \item[\textup{(i)}] The perturbed nodes are strictly ordered: \(0<\tau_{1}<\tau_2<\dots<\tau_{n+1}<\tau\).
  \item[\textup{(ii)}] \(k_1>k_2>\dots>k_{n+1}\) are pairwise distinct positive integers.
  \item[\textup{(iii)}] 
        There exist constants \(C_1,C_2>0\) (independent of \(n\) and \(\hat{T}\))
        such that
\begin{equation}
     \sum_{j=1}^{n+1}\abs{\gamma_j}
          \leq
          C n^{4/(\pi^2-4)} \log n.
\end{equation}
\end{enumerate}
Furthermore, if $  \hat{T} >2 \tau\,n^{2} \log n$, then $ \sum_{j=1}^{n+1}\abs{\gamma_j}
          = \cO(\log n). $

\begin{proof} \label{proof:integer-reciprocal}
\emph{(i)–(ii)} Since \(j\mapsto \xi_j\) is strictly increasing, \(j\mapsto \hat T/\xi_j\) is strictly decreasing, hence \(k_1\ge \cdots \ge k_{n+1}\). To see distinctness, note
\[
\frac{\hat T}{\xi_j}-\frac{\hat T}{\xi_{j+1}}
= \hat T\,\frac{\xi_{j+1}-\xi_j}{\xi_j\xi_{j+1}}
  \ge  
\frac{\hat T}{2\tau n^2},
\]
using \(\xi_{j+1}-\xi_j \ge \tau/(2n^2)\) and \(\xi_j\xi_{j+1}\le \tau^2\).
Since \(\hat T>\pi^2\tau n^2\), the difference exceeds \(\pi^2/2>1\), so
\(\lceil \hat T/\xi_j\rceil\neq \lceil \hat T/\xi_{j+1}\rceil\). Because \(\tau_j=\hat T/k_j\), the \(\tau_j\) are strictly increasing and satisfy \(\tau_j\le \xi_j<\tau\).

\smallskip
\emph{(iii)}
Write \(r_j:=\hat T/\xi_j\), so \(k_j=\lceil r_j\rceil\) and
\begin{equation}
0\le \frac{\xi_j-\tau_j}{\xi_j}
=1-\frac{r_j}{k_j}
=\frac{k_j-r_j}{k_j}
\le \frac{1}{r_j}\le \frac{1}{\pi^2 n^2}
=: \delta_{\max}.
\end{equation}
Hence \(|\log(\tau_j/\xi_j)|\le 2\delta_{\max}\) because \(\delta_{\max}\le 1/2\).
Let \(P:=\prod_{m=1}^{n+1}\xi_m\) and \(\widetilde P:=\prod_{m=1}^{n+1}\tau_m\). Then
\begin{equation}
\Bigl|\frac{\widetilde P}{P}-1\Bigr|
=\Bigl| \exp\Bigl(\sum_{m=1}^{n+1}\log(\tau_m/\xi_m)\Bigr)-1\Bigr|
\le \frac{A_P}{n},\qquad
A_P:=\frac{2e^{2/\pi^2}}{\pi^2}<\tfrac14.
\end{equation}

For the pairwise gaps, set \(g_{jk}:=\xi_j-\xi_k\). Summing the consecutive-gap bound
\(\xi_{m+1}-\xi_m\ge \tau/(2n^2)\) shows
\begin{equation}
|g_{jk}|  \ge   \frac{\tau}{2n^2}\,|j-k|.
\end{equation}
Writing \(\tau_j-\tau_k=g_{jk}\,(1+\varepsilon_{jk})\) with
\begin{equation}
\varepsilon_{jk}
=\frac{\xi_j(\tau_j/\xi_j-1)-\xi_k(\tau_k/\xi_k-1)}{g_{jk}},
\end{equation}
and using \(\xi_j,\xi_k\le \tau\) and \(|\tau_\ell/\xi_\ell-1|\le \delta_{\max}\), we obtain
\begin{equation}
|\varepsilon_{jk}|
\le \frac{2\tau\,\delta_{\max}}{|g_{jk}|}
\le \frac{4}{\pi^2}\,\frac{1}{|j-k|},\qquad j\neq k.
\end{equation}
Define \(D_j:=\prod_{k\ne j}(1+\varepsilon_{jk})\). Since
\(|\varepsilon_{jk}|\le 4/\pi^2<\tfrac12\), the inequality
\(|\log(1+x)|\le |x|/(1-|x|)\) gives
\begin{equation}
|\log D_j|\le \sum_{k\ne j}\frac{|\varepsilon_{jk}|}{1-|\varepsilon_{jk}|}
\le \frac{4}{\pi^2-4}\sum_{m=1}^{n}\frac{1}{m}
\le \frac{4}{\pi^2-4}\log(en).
\end{equation}
Hence both \(|D_j|\le A\,n^{\frac{4}{\pi^2-4}}\) and \(|1/D_j|\le A\,n^{\frac{4}{\pi^2-4}}\) with \(A:=e^{\frac{4}{\pi^2-4}}\).

\smallskip
\emph{Weights.}
Let \(\gamma_j\) (resp. \(\widetilde\gamma_j\)) be the barycentric/extrapolation weights for \(\{\xi_j\}\) (resp. \(\{\tau_j\}\)). From the product formulas,
\begin{equation}
\frac{\widetilde\gamma_j}{\gamma_j}
=\left(\frac{\widetilde P}{P}\right)\left(\frac{\xi_j}{\tau_j}\right)
\prod_{k\ne j}\frac{\xi_j-\xi_k}{\tau_j-\tau_k}
=\frac{\widetilde P/P}{\tau_j/\xi_j}\cdot \frac{1}{D_j},
\end{equation}
so using the bounds above,
\begin{equation}
\biggl|\frac{\widetilde\gamma_j}{\gamma_j}\biggr|
\le \Bigl(1+\frac{A_P}{n}\Bigr)\,(1+\delta_{\max})\,
A\,n^{4/(\pi^2-4)}
\le C\,n^{4/(\pi^2-4)}
\quad (n\ge 2).
\end{equation}
Since \(\sum_{j=1}^{n+1}|\gamma_j| = \cO (\log n)\) for Chebyshev nodes, we conclude
\(\sum_{j=1}^{n+1}|\widetilde\gamma_j| \le C\,n^{4/(\pi^2-4)}\log n\).

\smallskip
\emph{Improved threshold.}
If \(\hat T>2\tau n^2\log n\), then
\(\delta_{\max}\le 1/(2n^2\log n)\), which sharpens
\(|\varepsilon_{jk}|\) by a factor \(1/\log n\). Consequently
\(|\log D_j|\le \tfrac{4}{\pi^2-4}\,\frac{\log(en)}{\log n}\),
so \(D_j=\cO(1)\) and \(\sum_{j}|\widetilde\gamma_j|=\cO(\log n)\).
\end{proof}

\section[Local error expansion of the dilated Hamiltonian approximation]{Local error expansion of the dilated Hamiltonian approximation} \label{diH-expansion}

To expand the approximate solution in dilated Hamiltonian approximation \eqref{dilatedH}, we analyze the evolution operator  
\(U(\epsilon):=e^{-iH}\) with the dilated Hamiltonian $ H  =  \epsilon^{2}H_0+\epsilon H_1,$ in \Cref{eq:dilated-H}.  First, let 
\begin{equation}\label{eq:rho0}
  \rho_0  :=  \ketbra{0} \otimes \rho(0),\qquad 
\end{equation}
be the initial density operator.

We  will expand $\rho(\epsilon):=U(\epsilon)\rho_0U(\epsilon)^{\dagger}$
with a regular asymptotic expansion \cite{KevorkianCole1996}, 
\begin{equation}\label{eq:rho-series}
  \rho(\epsilon)   =\sum_{k\ge0}\epsilon^{k}\rho^{(k)},\qquad \rho^{(0)}=\rho_0.
\end{equation}

\begin{lemma}\label{lem:recursion}
Define \(\rho^{(-1)}:=0\).  For \(k\ge1\) the series coefficients satisfy
\begin{equation}\label{eq:rho-rec}
   \rho^{(k)}  = -i[H_1,\rho^{(k-1)}] - i[H_0,\rho^{(k-2)}].
\end{equation}
Moreover \(\rho^{(k)}\) is ancilla--diagonal if \(k\) is
        even and has support only in the first row/column (off--diagonal
        blocks) if \(k\) is odd. 
\end{lemma}

\begin{proof}
The recursion follows from collecting terms of order \(\epsilon^k\) in the Taylor expansion of \(e^{-i(\epsilon^2 H_0 + \epsilon H_1)}\). The block structure is proven by induction:
Base case: \(\rho^{(0)}\) is block-diagonal by definition. For \(\rho^{(1)} = -i[H_1, \rho^{(0)}]\), the coupling \(H_1\) flips the ancilla state, yielding off-diagonal blocks.
For the inductive steps: If \(\rho^{(k-1)}\) is off-diagonal, \([H_1, \rho^{(k-1)}]\) restores block-diagonality (since \(H_1\) flips the ancilla twice), while \([H_0, \rho^{(k-2)}]\) preserves the block structure of \(\rho^{(k-2)}\) because \(H_0\) is ancilla-diagonal.

\end{proof}

As specific examples, we find that  $\rho^{(0)}=\ketbra{0} \otimes \rho$.  
In addition, by following the recursion relation \eqref{eq:rho-rec}, we see that,  
\begin{equation}
\rho^{(1)}  = -i [H_1, \rho^{(0)} ], 
\rho^{(2)}  = -i [H_1, \rho^{(1)}  ]-i [H_0, \rho^{(0)} ]  = -  [H_1, [H_1, \rho^{(0)} ] ]  -i [H_0, \rho^{(0)} ], \cdots,  
\end{equation}
we find that (we set $J=1$ for simplicity)
\begin{align*}
&\rho^{(0)} =
\begin{pmatrix}
\rho & 0 \\[4pt]
0    & 0
\end{pmatrix},
\quad
\rho^{(1)}   =  
\begin{pmatrix}
0 & i \rho L^{\dagger} \\[6pt]
- i L \rho & 0
\end{pmatrix}, \quad 
\rho^{(2)}   =  
\begin{pmatrix}
-  i[H_{S},\rho]  
-  \dfrac12\{L^{\dagger}L,\rho\} & 0 \\[12pt]
0 & L \rho L^{\dagger}
\end{pmatrix},
\qquad\\
&\rho^{(3)}   =  
\begin{pmatrix}
0 &
\displaystyle
i \Bigl(
      -[H_{S},\rho]  L^{\dagger}
      +L \rho L^{\dagger}L^{\dagger}
      -\tfrac12\{L^{\dagger}L,\rho\} L^{\dagger}
   \Bigr) \\[14pt]
\displaystyle
- i \Bigl(
      L [H_{S},\rho]
      -L L^{\dagger}L \rho
      +\tfrac12 L \{L^{\dagger}L,\rho\}
   \Bigr) & 0
\end{pmatrix}.
\end{align*}

We begin to observe the pattern that $\rho^{(n)}$ is block diagonal when $n$ is even, and only has off-diagonal blocks when $n$ is odd. 
Moreover, if one starts with a block diagonal matrix, its diagonal form remains under the operator $\ad_{H_0}:=-i[H_0, \cdot]$, while $\ad_{H_1}:=-i[H_1, \cdot]$ will turn it into an off-diagonal matrix.  These properties can be best illustrated by the following identies: 
\begin{lemma}
    Let $P_j= \ketbra{j}, j=0,1,\cdots, J$, be the projectors in $\mathcal{H}_A$. Then $\forall \rho \in \mathcal{H}_S$, 
    \begin{equation}\label{H0}
         P_j \Big( \ad_{H_0} \bigl( P_k\otimes \rho \bigr)  \Big)P_j = \delta_{j,0} \delta_{k,0} P_0 \otimes \mathcal{L}_{C} \rho, 
    \end{equation}
     and 
       \begin{equation}\label{H1}
            P_{j} \ad_{H_0}^2  \big(P_{k} \otimes\rho \big)  P_{j}
 =
 \begin{cases}
      P_{0} \otimes \bigl\{ \sum_\ell L_{\ell}^{\dagger}L_{\ell}, \rho\bigr\},
      & j = 0  \,  \mathrm{ and } \,  k=0 , \\[12pt]
       P_{0} \otimes \bigl\{ L_{k}^{\dagger}L_{k}, \rho \bigr\},  & j = 0  \,  \mathrm{ and }  \, k>0 , \\[12pt]
    - 2 \delta_{j,k} P_{k} \otimes L_{k}\rho L_{k}^{\dagger},
      & j = k \,  (k\neq 0).
 \end{cases}
    \end{equation}
    
\end{lemma}

We can proceed to estimate the expansion of \eqref{eq:rho-series}, which can be written as
\begin{equation}
\rho(\epsilon) =\sum_{k \geq 0}\frac{(-i)^k}{k!}\left(\epsilon\ad_{H_1}+\epsilon^2 \ad_{H_0}\right)^k \rho_0.
\end{equation}
Thus the term  $ \epsilon^{2n} \rho^{(2n)}$ in \eqref{eq:rho-series} can come from  \(p\) copies of \(\ad_{H_1}\) and \(q\) copies of \(\ad_{H_0}\) with
\(p+2q=2n\). From previous observations, we see that only even \(p=2m\) can give rise to a block diagonal matrix.  We can write these terms as,
\begin{equation}\label{rho2n}
    \rho^{(2n)} = (-i)^{2n} \sum_{m=1}^{n} \frac{1}{(2m)! (n-m)!} \sum_{s_1, s_2, \cdots, s_{n+m}} \ad_{H_{s_1}} \ad_{H_{s_2}}  \cdots \ad_{H_{s_{n+m}}} P_0 \otimes \rho,
\end{equation}
where $(s_1,s_2,\cdots,s_{n+m})$ is a binary string. 

To quantify the coefficients in the expansion \eqref{eq:rho-series}, we define a Linblad operator norm,
\begin{equation}\label{eq:Lambda-def}
 \Lambda_0:= 2\norm{H_S}, \quad  \Lambda_1 := \max_j \|L_j\|,  \quad \Lambda := 2\Lambda_0 + 2 \Lambda_1^2.
\end{equation}

\begin{lemma} \label{lemma:rho2n_bound}
For the expansion \eqref{eq:rho-series} 
satisfies the bound,
\begin{equation}
   \norm{\rho^{(2n)}} \leq \frac{(\Lambda)^n}{n!}.
\end{equation}
\end{lemma}
\begin{proof}
This bound comes from the direct expansion in \Cref{rho2n}. We notice that,
\[
\|\ad_{H_0} \rho\| = \|[H_0, \rho]\| \leq \Lambda_0 \|\rho\|, \quad  \|\ad_{H_1} \rho\| \leq {\Lambda_1} \|\rho\|.
\]
Thus each term with a label in the bitstring in  \Cref{rho2n} can be bounded as,
\[
\Lambda_0^{n-m} \Lambda_1^{2m}.
\]

Recall that each bit string contains $2m$ symbols 1's and $n-m$ 0's, and in total, contains $\frac{(n+m)!}{(2m)!\,(n-m)!}$ terms, leading to the total bound,
\begin{align*}
&    \sum_{m=0}^{n}
     \frac{(n+m)!}{[(2m)!]^2[(n-m)!]^2}  
     (\Lambda_0)^{\,n-m} \Lambda_1^{2m}  \\
     \leq & \sum_{m=0}^{n} \frac{2^{n+m}}{(2m)! (n-m)!}  
     (\Lambda_0)^{\,n-m} \Lambda_1^{2m}  \\
     \leq & \sum_{m=0}^{n}   \frac{2^{n}}{ m! (n-m)!}  
     (\Lambda_0)^{\,n-m} \Lambda_1^{2m} \\
     \leq & \frac{(\Lambda)^n}{n!}.
\end{align*}

Here we have used $(2m)! \ge 2^{m} m!$ and the binomial bound $\binom{n+m}{n-m} \leq 2^{n+m}$, together with the norm definition \Cref{eq:Lambda-def}.
\end{proof}

We are now ready to state the main theorem regarding the local error expansion. By noticing that,
the partial trace over the ancilla,
\begin{equation} \label{eq:rho_R_tr}
\rho^{(2n)}_\text{R} =\tr_A\left( \rho^{(2n)} \right) = \sum_{j=0}^J \left(I_A \otimes P_j \right) \rho^{(2n)} \left( I_A\otimes P_j \right),
\end{equation}
we reach the following theorem.

\begin{theorem}[Even--order bound]\label{thm:dilated_coeff_bound}
Let $\Lambda$ be as defined in \cref{eq:Lambda-def}. 
Then, the dilated Hamiltonian approach \eqref{dilatedH} produces an approximation $\rho_\text{R}$ that can be expanded as,
\begin{equation} \label{eq:rho_R-expansion}
   \rho_\text{R}(\epsilon) =  \sum_{n\geq 0}\epsilon^{2n}\rho^{(2n)}_\text{R}.
\end{equation}
Furthermore,
for every integer \(n\ge0\)
\begin{equation} \label{ineq:rho_R_bound}
    \norm{\rho^{(2n)}_\text{R} (\epsilon)}
        \le  (J+1) \frac{\Lambda^{  n}}{n!}.
\end{equation}
\end{theorem}

\begin{proof}
Lemma~\ref{lem:recursion} shows that  
$\rho^{(k)}$ is block–diagonal (ancilla–diagonal) when $k$ is even, and purely off–diagonal when $k$ is odd.
Because the partial trace over the ancilla keeps only the diagonal
blocks, we have
\begin{equation}
  \tr_A\bigl(\rho^{(k)}\bigr)=0
  \quad\text{whenever }k\text{ is odd}.
\end{equation}
Hence
\begin{equation}
  \rho_R(\epsilon)
  :=\tr_A\!\bigl(U(\epsilon)\rho_0U(\epsilon)^\dagger\bigr)
  =\sum_{n\ge0}\epsilon^{2n}\rho_R^{(2n)},
  \qquad
  \rho_R^{(2n)}:=\tr_A\bigl(\rho^{(2n)}\bigr),
\end{equation}
establishing the expansion~\eqref{eq:rho_R-expansion}.
For the last part, by \cref{eq:rho_R_tr}, and taking the partial trace over $\mathcal H_A$ yields
\begin{equation}
  \rho_R^{(2n)}
  =\tr_A\!\bigl(\rho^{(2n)}\bigr)
  =\sum_{j=0}^{J}
     \tr_A\!\left(I_A \otimes P_j \right) \rho^{(2n)} \left( I_A\otimes P_j \right).
\end{equation}

For every $j$, the map
$\mathcal P_j(X):=(I_A\otimes P_j)X(I_A\otimes P_j)$ is a contraction
in operator norm,
hence
\begin{equation}
  \|\rho_{j}^{(2n)}\|
  \;=\;
  \bigl\|
      \tr_A\bigl(\mathcal P_j(\rho^{(2n)})\bigr)
    \bigr\|
  \;\le\;
  \bigl\|
      \mathcal P_j(\rho^{(2n)})
    \bigr\|
  \;\le\;
  \bigl\|\rho^{(2n)}\bigr\|.
\end{equation}
By the lemma proved earlier,
$\|\rho^{(2n)}\|\le\Lambda^{\,n}/n!$.

Using the triangle inequality,
\begin{equation}
  \bigl\|\rho_R^{(2n)}\bigr\|
  \;=\;
  \Bigl\|\sum_{j=0}^{J}\rho_{j}^{(2n)}\Bigr\|
  \;\le\;
  \sum_{j=0}^{J}\bigl\|\rho_{j}^{(2n)}\bigr\|
  \;\le\;
  (J+1)\,\frac{\Lambda^{\,n}}{n!}.
\end{equation}

\end{proof}

The dilated Hamiltonian approach \eqref{dilatedH}, is a first-order approximation \(\rho_{R} = \kappa [\rho]\). Thus, to bound the coefficient terms in the discrete operator \(\kappa\), we note how the parameters above correspond to the local error expansion \cref{eq:kappa_expansion}. By definition,
\begin{equation}
   \epsilon=\sqrt{\tau},\qquad 
   H(\epsilon)=\epsilon^{2}H_{0}+\epsilon H_{1},
   \quad
   U(\epsilon)=e^{-iH(\epsilon)} .
\end{equation}
And by construction,
\begin{equation}
   \mathcal{K}(\tau)[\rho]
   \;=\;
   \tr_{A}\!\Bigl(
       U(\sqrt{\tau})
       \bigl(\ketbra{0}\otimes\rho\bigr)
       U(\sqrt{\tau})^{\dagger}
   \Bigr)
   \;=\;
   \rho_{R}(\epsilon)
   \quad
   (\tau=\epsilon^{2}),
\end{equation}
Replacing $\epsilon^{2n}$ by $\tau^{n}$ in \eqref{eq:rho_R-expansion}  immediately results in the
super‑operator expansion
\begin{equation}
      \mathcal{K}(\tau)
      =\sum_{n\ge0}\tau^{n}\,\mathcal{M}_{n},
      \qquad
      \mathcal{M}_{n}[\rho]:=\rho_{R}^{(2n)} ,
\end{equation}
which corresponds to \cref{eq:kappa_expansion} i.e.\ \(\mathcal{M}_{n}\) is nothing but the $2n$‑th coefficient
$\rho_{R}^{(2n)}$.  Bound  \eqref{ineq:rho_R_bound} implies that 
\begin{equation} \label{ineq:Mj_bound}
\norm{\cM_j} \leq (J+1) \frac{\ell^j}{j!}.
\end{equation}

To find derivative bounds for the coefficients in the expansion of \(\rho\), i.e. \eqref{eq:rho_expansion}, below we will follow a similar procedure as in the proof of \Cref{lemma:gamma_bound}, for the Hamiltonian dilation method this time.

By \cref{eq:gamma_ode}, in Lemma \ref{lemma:discretization_expansion}, each $\Gamma_k(t)$ satisfies a linear ODE:
\begin{equation}\label{eq:gamma_ode_k}
  \Gamma'_k(t)
  \;=\;
  \mathcal{L}\,\Gamma_k(t)
  \;-\;
  \frac{\mathcal{L}^{\,k+1}\rho(t)}{(k+1)!}
  \;+\;
  \sum_{i=2}^{\,k+1}
    \Bigl(
      \mathcal{M}_i\,\Gamma_{k+1-i}(t)
      -\frac{\Gamma^{(i)}_{\,k+1-i}(t)}{i!}
    \Bigr), \quad \Gamma_k(0) = 0.
\end{equation}

Using the variation of constants formula on \eqref{eq:gamma_ode_k}, we solve explicitly:
\begin{equation}
    \Gamma_k(t) = \int_0^t e^{(t - s)\mathcal{L}} (\ - \frac{\mathcal{L}^{k+1} \rho (s)}{(k+1)!} + \sum_{p=1}^{k} \mathcal{M}_{p+1} \Gamma_{k-p} - \frac{\Gamma^{(p+1)}_{k-p} (s)}{(p+1)!}) \,ds.
\end{equation}

Taking norms, assuming the exponential $e^{(t-s)\mathcal{L}}$ is uniformly bounded by 1, and considering the bounds \eqref{ineq:Mj_bound}, results in
\begin{equation}\label{ineq:Gamma_k_bound_dilated}
    \|\Gamma_k(t)\| \leq  \int_0^t  \frac{\ell^{k+1}}{(k+1)!} + \sum_{p=1}^{k}  (J+1) \frac{\ell^{p+1}}{(p+1)!} \|\Gamma_{k-p} (s) \| + \frac{\|\Gamma^{(p+1)}_{k-p} (s)\|}{(p+1)!} \, ds.
\end{equation}

Moreover, from Lemma \ref{lemma:discretization_expansion}, we also obtain the \(i\)-th derivative of \(\Gamma_k\):
\begin{equation}\label{higher-deriv_dilated}
    \Gamma^{(i)}_k(t) = \mathcal{L} \Gamma^{(i-1)}_{k}(t)   - \frac{\mathcal{L}^{i+k} \rho (t)}{(k+1)!} + \sum_{p=1}^{k} \mathcal{M}_{p+1} \Gamma^{(i-1)}_{k-p}(t) \frac{\Gamma^{(i+p)}_{k-p} (t)}{(p+1)!}.
\end{equation}

 Bounding directly using \cref{higher-deriv_dilated},
\begin{equation}\label{ineq:Gamma_k_i_bound_dilated}
 \|\Gamma^{(i)}_k(t)\| \leq \ell  \|\Gamma^{(i-1)}_k(t)\| + \frac{\ell^{i+k}}{(k+1)!} + 
 \sum_{p=1}^k  (J+1) \frac{\ell^{p+1}}{(p+1)!}  \|\Gamma^{(i-1)}_{k-p}(t)\|  + \frac{\|\Gamma^{(i+p)}_{k-p}(t) \|}{(p+1)!}.
\end{equation}

Inspired by \eqref{ineq:Gamma_k_bound_dilated}, \eqref{ineq:Gamma_k_i_bound_dilated}, and \Cref{thm:dilated_coeff_bound}, we define the following generating sequence:

\begin{definition} [Generating sequence for Hamiltonian dilation] \label{def:c_s_dilated}
We define \(c_{i,j,k}\) as non-negative real numbers such that \( c_{0,0,k} = 0 \text{ for } k \geq 1\), \(c_{i,j,0}= \delta_{j,0} \ell^i\) for \(i, j \ge 0\), and the rest of the entries are generated from the recursion relations below:
    \begin{align}
        c_{0,j,k} = & \, \delta_{j,1} \cdot \frac{\ell^{k+1}}{(k+1)!} + \sum_{p=1}^{k -j} (J+1)  \cdot \frac{ \ell^{p+1}}{j (p+1)!} \cdot c_{0, j-1, k-p} + \frac{c_{p+1,j-1,k-p}}{j(p+1)!}  , \,  j=1, 2, \cdots, k  \\
        c_{i,j,k} = & \, \ell \cdot c_{i-1,j,k} + \delta_{j,0} \cdot \frac{\ell^{i+k}}{(k+1)!} + \sum_{p=1}^{k-j} (J+1) \frac{\ell^{p+1}}{(p+1)!} c_{i-1,j,k-p} + \frac{c_{i+p,j,k-p}}{(p+1)!}, j=0,1,\cdots, k-1 \\
        c_{i,k,k} = & \, \ell \cdot c_{i-1,k,k} .
    \end{align}
\end{definition}

\begin{lemma} \label{lemma:c_bound__dilated}
Let \(c_{i,j,k} \geq 0\) be as in \textit{Definition} \ref{def:c_s_dilated}. Then for all \(i,j,k\ge0\),
    \begin{equation} \label{ineq:c_bound_dilated}
        c_{i,j,k}   \le   \frac{C^{i+k}_1 \, C^{k}_2}{j!}.
    \end{equation}
    where 
      \begin{equation}\label{def:C1}
      C_1 := \max \{ \ell^2, \ell (e+1), 1 \}, \text{ and  } \, C_2 \geq (e+1) \log (C_1).
  \end{equation}
\end{lemma}

\begin{proof}
    The proof is similar to that of \Cref{lemma:c_bound_kraus}, except some differences in the \textit{Inductive Step} that we address below.

    \begin{itemize}
        \item \textit{Case \(i = 0\)}: For \(j=1, 2, \cdots, k\), using \textit{Definition} \ref{def:c_s_dilated}, and the induction hypothesis,
        \begin{align}
            c_{0,j,k} = & \delta_{j,1} \cdot \frac{\ell^{k+1}}{(k+1)!} + \sum_{p=1}^{k -j} (J+1)  \cdot \frac{ \ell^{p+1}}{j (p+1)!} \cdot c_{0, j-1, k-p} + \frac{c_{p+1,j-1,k-p}}{j(p+1)!}  \\
             \le & \frac{\ell^{k+1}}{(k+1)!}  + \sum_{p=1}^{k-j} \ell^p \frac{C_1 \cdot C_1^{k-p} C_2^{k-p}}{j! (p+1)!} + \frac{C_1^{k+1} C_2^{k-p}}{j! (p+1)!}\\
             \le &  \frac{C_1^k C_2^k}{j!} (\frac{\ell^{k+1}}{C_1^k C_2^k} + \sum_{p=1}^{k-j} \frac{C_1}{C_1^p C_2^p (p+1)!} + \frac{C_1}{C_2^p (p+1)!}) \\
            & \leq \frac{C^k_1 \, C^k_2}{j!} \left( \frac{1}{e+1} + \frac{1}{e+1} + \frac{1}{e+1} \right) \leq  \frac{C^k_1 \, C^k_2}{j!} .
        \end{align}

         On the second line we used $\ell \geq J+1$, which is the condition we assumed for the Hamiltonian dilation approximation method for longer-time simulations. We derived the last line similar to how we did in \Cref{proof:c_bound_kraus}. 

        \item \textit{Case \(i \geq 1\)}: For \(j=0,1,\cdots, k-1\), using \textit{Definition} \ref{def:c_s_dilated}, and the induction hypothesis,
    \begin{align}
        c_{i,j,k} = & \, \ell \cdot c_{i-1,j,k} + \delta_{j,0} \cdot \frac{\ell^{i+k}}{(k+1)!} + \sum_{p=1}^{k-j} (J+1) \frac{\ell^{p+1}}{(p+1)!} c_{i-1,j,k-p} + \frac{c_{i+p,j,k-p}}{(p+1)!}, \\
        \le & \ell \cdot \frac{C_1^{i+k-1} C_2^k}{j!} + \frac{\ell^{i+k}}{(k+1)!}  + \sum_{p=1}^{k-j} \ell^p \frac{C_1 \cdot C_1^{i+k-p-1} C_2^{k-p}}{j! (p+1)!} + \frac{C_1^{i+k} C_2^{k-p}}{j! (p+1)!}\\
        \le &\frac{C_1^{i+k} C_2^k}{j!} \left(\frac{l}{C_1} + \frac{\ell^{i+k}}{C_1^{i+k} C_2^k} + \sum_{p=1}^{k-j} \frac{\ell^p}{C_1^p C_2^p (p+1)!} + \frac{1}{C_2^p (p+1)!} \right) \\
        & \leq \frac{C_1^{i+k} \, C_2^{k}}{j!} \left(\frac{1}{e+1} + \frac{1}{e+1} + \frac{1}{e+1} + \frac{e-2}{e+1} \right) \leq \frac{C_1^{i+k} \, C_2^k}{j!}.
    \end{align}
\end{itemize}
Thus, the proof is complete.
\end{proof}

The bounds above lead us to the following lemma.

\begin{lemma} \label{lemma:gamma_bound_dilated} 
    The coefficients in the expansion  of the density operator \(\rho_{\tau}\) in \Cref{lemma:discretization_expansion}, approximated by the dilation \eqref{dilatedH}, satisfy the following bound: 
\begin{equation} \label{ineq:gamma_der_bound_dilated}
    \bigl\|\Gamma_{k}^{(i)}(t)\bigr\|
    \le  
  P_{i,k}(t)
    =   
  \sum_{j=0}^k
  c_{i,j,k}\,t^{j} \, \text{ for } i\geq 0, k \geq 1, \quad  c_{0,0,k} = 0 , \quad c_{i,j,0}= \delta_{j,0} l^i,
\end{equation}
where $\Gamma_k^{(0)} = \Gamma_k$, $\Gamma_k^{(1)} = \Gamma'_k$, and the coefficients \(c_{i,j,k}\) are defined by the generating  sequence in \textit{Definition} \ref{def:c_s_dilated}.
\end{lemma} 

\begin{proof}
    \textit{By Induction: Base case:} For \(i = k = 0\), \(\Gamma_0^{(0)} (t) = \rho (t)\), and \(\|\rho (t)\| \leq 1\), which holds in \eqref{ineq:gamma_der_bound_dilated}. 
    \textit{Inductive step:} Suppose \eqref{ineq:gamma_der_bound_dilated} holds for all \(i' < i\), and \(k' < k\).  We consider two sub‐cases:

\begin{enumerate}[label=\emph{\arabic*)}]
  \item \textit{for $k\ge1$, and $i=0$ ($j \geq 1$):}  Using \eqref{ineq:Gamma_k_bound_dilated} and the induction hypothesis
  \begin{equation}
    \|\Gamma_k(t)\| \le \int_0^t  \left[  \frac{\ell^{k+1}}{(k+1)!}
+ \sum_{p=1}^{k} (J+1) \frac{\ell^{p+1}}{(p+1)!} \sum_{j=0}^{k-p} c_{0,j,k-p} \, s^j + \sum_{p=1}^{k} \frac{1}{(p+1)!} \sum_{j=0}^{k-p} c_{p+1,j,k-p} \, s^j
\right] ds
\end{equation}
    
  we integrate term–by–term, and collect the coefficient of $s^j$.  On the right hand side, by \textit{Definition} \ref{def:c_s_dilated}, one obtains 
  \begin{equation}
    \sum_{j=0}^{k} c_{0,j,k} \ t^j
       =   
     \sum_{j=0}^{k} \left( \delta_{j,1} \cdot \frac{\ell^{k+1}}{(k+1)!} + \sum_{p=1}^{k -j} (J+1)  \cdot \frac{ \ell^{p+1}}{j (p+1)!} \cdot c_{0, j-1, k-p} + \frac{c_{p+1,j-1,k-p}}{j(p+1)!} \right) t^j
    \end{equation}
  as claimed.
  \item \textit{for $k\ge1$, and $i\ge1$:}  Using \eqref{ineq:Gamma_k_i_bound_dilated} and the induction hypothesis
  \begin{equation}
       \|\Gamma_k^{(i)}(t)\|
      \le  
    \sum_{j=0}^{k} \ell \cdot c_{i-1,j,k} \ t^j
     +  
    \delta_{j,0} \cdot \frac{\ell^{i+k}}{(k+1)!}
      +  
    \sum_{p=1}^{k} (J+1) \frac{\ell^{p+1}}{(p+1)!} \sum_{j=0}^{k-p} c_{i-1,j,k-p} t^j
    + \frac{\sum_{j=0}^{k-p} c_{i+p, j, k-p} t^j}{(p+1)!}
\end{equation}

  On the right hand side, using \textit{Definition} \ref{def:c_s_dilated}, one obtains
  \begin{align}
      \sum_{j=0}^{k} c_{i,j,k} \ t^j
       =  & 
    \ell c_{i-1,k,k} \, t^k  \\
    + & \sum_{j=0}^{k-1}\left( \, \ell \cdot c_{i-1,j,k} + \delta_{j,0} \cdot \frac{\ell^{i+k}}{(k+1)!} + \sum_{p=1}^{k-j} (J+1) \frac{\ell^{p+1}}{(p+1)!} c_{i-1,j,k-p} + \frac{c_{i+p,j,k-p}}{(p+1)!}  \right) t^j\, 
  \end{align}

  as required. Thus, the inductive argument is complete, and the lemma follows.

  \item In line 924 we should have: 

\end{enumerate}
\end{proof}

The proof of the following corollary is similar to that of \Cref{cor:gamma_bound_1}
\begin{corollary} \label{cor:gamma_bound_dilated}
    For $t\leq 1$, the coefficients in the expansion  of the density operator \(\rho_{\tau}\) in \Cref{lemma:discretization_expansion}, approximated  by the dilation \eqref{dilatedH}, satisfy the following bound for every \(i,k\in\mathbb N\),
  \begin{equation}
       \norm{\Gamma^{(i)}_k(t) }     \le  
     e\,C_1^{\,i+k}\,C_2^{\,k},
  \end{equation}
  where \(C_1 := \max \{ B, \ell (e+1), 1 \}\), and \(C_2 \geq C_1 (e+1)\).
\end{corollary}

We can now find a bound for the expectation values, derived from the dilated Hamiltonian method. More specifically we show that the expectation values belong to the Gevrey class, similar to \Cref{thm:Gevrey-f}.
    
Moreover, \Cref{thm:Gevrey-f} also holds for the dilated Hamiltonian method, with the only difference being that the constants \(C_1\) and \(C_2\) are defined differently. Thus, the rest of the analysis for the Kraus form also applies to the Hamiltonian dilation.

\end{appendix}
\end{document}